\documentclass[11pt,a4paper]{article}
\usepackage{jheppub}
\usepackage[utf8]{inputenc}
\usepackage[normalem]{ulem}
\usepackage[title]{appendix}
\makeatletter
\long\def\emph#1{\ifmmode\nfss@text{\em #1}\else\hmode@bgroup\text@command{#1}\em\check@icl #1\check@icr\expandafter\egroup\fi}
\makeatother
\usepackage{epsfig}
\usepackage{graphicx}
\usepackage{amsfonts,amsmath,amsthm,amssymb}
\usepackage{epiolmec}
\usepackage{phaistos}
\usepackage{adjustbox}
\newtheorem{theorem}{Theorem}
\newtheorem{corollary}{Corollary}
\newtheorem{definition}{Definition}
\newtheorem{notation}{Notation}
\newtheorem{remark}{Remark}
\newtheorem{proposition}[]{Proposition}

\numberwithin{equation}{section}
\newcommand{\tr}{\operatorname{tr}}
\usepackage{ulem}
\usepackage{multicol}
\usepackage{latexsym}
\usepackage{mathrsfs} 
\newcommand\Vtextvisiblespace[1][.5em]{%
  \mbox{\kern.06em\vrule height.3ex}%
  \vbox{\hrule width#1}%
  \hbox{\vrule height.3ex}}
\usepackage{cancel}

\usepackage{xparse}
\usepackage{float}
\usepackage{todonotes}
\usepackage{etoolbox}
\usepackage{tikz}
\usepackage{bbold}
\usepackage{upgreek}
\usetikzlibrary{arrows,positioning,cd,calc,math,decorations.pathreplacing,decorations.markings,decorations.pathmorphing,patterns
}
\tikzset{wave/.style={decorate, decoration=snake}}

\usepackage{xcolor}
\definecolor{MyBlue}{rgb}{0.25,0.5,0.75}
\colorlet{NextBlue}{MyBlue!20}
\colorlet{SecondBlue}{MyBlue!40}
\usepackage{subcaption}
\NewDocumentCommand{\tens}{t_}
 {%
  \IfBooleanTF{#1}
   {\tensop}
   {\otimes}%
 }
\NewDocumentCommand{\tensop}{m}
 {%
  \mathbin{\mathop{\otimes}\displaylimits_{#1}}%
 }

\newcommand{\eq}[1]{\begin{equation}\begin{gathered}#1\end{gathered}\end{equation}}
\newcommand{\eqs}[1]{\begin{equation}\begin{gathered}\begin{split}#1\end{split}\end{gathered}\end{equation}}

\DeclareMathOperator{\res}{Res}

\newcommand{\bs}{\boldsymbol}

\newcommand{\cP}{\mathcal{P}}
\newcommand{\cin}{\mathcal{C}_{in}}
\newcommand{\cout}{\mathcal{C}_{out}}
\newcommand{\cC}{\mathcal{C}}
\newcommand{\cH}{\mathcal{H}}

\newcommand{\Yt}{\Psi_{3pt}}

\newcommand{\T}{\mathcal{T}}

\newcommand{\diag}{\text{diag}}

\newcommand{\dd}{{\rm d}}

\newcommand{\cA}{\mathcal{A}}
\newcommand{\cM}{\mathcal{M}}
\usepackage{tikz}
\usepackage{pgfplots}
\usepackage{upgreek}
\usetikzlibrary{decorations.markings,pgfplots.fillbetween}
\definecolor{darkspringgreen}{rgb}{0.05, 0.5, 0.06}
\definecolor{MyBlue}{rgb}{0.25,0.25,0.75}
\definecolor{MyRed}{rgb}{0.75,0.25,0.25}
\colorlet{NextBlue}{MyBlue!20}
\colorlet{SecondBlue}{MyBlue!40}

\title{Monodromy dependendence and symplectic geometry of isomonodromic tau functions on the torus}
\author[a,b]{Fabrizio Del Monte,}
\author[c]{Harini Desiraju,}
\author[d,e,f]{Pavlo Gavrylenko}
\affiliation[a]{Centre de Recherches Math\'ematiques (CRM), Universit\'e de Montr\'eal C. P. 6128, succ. Centre Ville, Montr\'eal, Qu\'ebec, Canada H3C 3J7}
\affiliation[b]{Department of Mathematics and Statistics, Concordia University, 1455 de Maisonneuve Blvd. W. Montr\'eal, QC H3G 1M8 Canada}
\affiliation[c]{University of Sydney,\\
F07 Carslaw building}
\affiliation[d]{Section de Math\'ematiques, Universit\'e de Gen\`eve, 1205 Geneva, Switzerland}
\affiliation[e]{Max-Planck-Institut f\"ur Mathematik, 53111 Bonn, Germany}
\affiliation[f]{Bogolyubov Institute for Theoretical Physics, Metrologichna 14-b, 03143 Kyiv, Ukraine}
\emailAdd{delmonte@crm.umontreal.ca}
\emailAdd{harini.desiraju@sydney.edu.au}
\emailAdd{pasha.145@gmail.com}
\abstract{We compute the monodromy dependence of the isomonodromic tau function on a torus with $n$ Fuchsian singularities and $SL(N)$ residue matrices by using its explicit Fredholm determinant representation. We show that the exterior logarithmic derivative of the tau function defines a closed one-form on the space of monodromies and times, and identify it with the generating function of the monodromy symplectomorphism. As an illustrative example, we discuss the simplest case of the one-punctured torus in detail. Finally, we show that previous results obtained in the genus zero case can be recovered in a straightforward manner using the techniques presented here.}
\begin{document}

\maketitle

\section{Introduction}

Isomonodromic tau functions are defined as the generating functions of Poisson commuting Hamiltonians $H_i$ that generate flows in times $t_i\in\mathbb{T}$,
\begin{equation}\label{eq:TauDefIntro}
	\dd_t\log\T:=\sum_{\mathrm{times }\, t_i}\dd t_i\partial_{t_i}\log\T=\omega_{JMU}:=\sum_{\mathrm{times }\, t_i}H_i\dd t_i.
\end{equation}
Here $\dd_t$ is the exterior differential on the space of times $\mathbb{T}$, $\omega_{JMU}\in T^*\mathbb{T}$ is the so called Jimbo-Miwa-Ueno (JMU) one-form \cite{JMU1981}, and its closedness
\begin{align}
\dd_t\,\omega_{JMU}=0
\end{align}
is equivalent to the consistency of the Hamiltonians flows, describing deformations of a linear system of ODEs
\begin{equation}\label{eq:LinSysIntro}
    \partial_z\Phi(z)=\Phi(z)A(z)
\end{equation}
that preserve its monodromies. The Hamiltonians themselves are obtained from contour integrals of $\frac{1}{2}\tr A^2(z)$. Through the Riemann-Hilbert correspondence that maps the space $\mathcal{A}$ of coefficients of the linear system to the space of monodromies $\mathcal{M}$, the Hamiltonians can be written in terms of the times and the monodromy data.
The tau function $\T$ in \eqref{eq:TauDefIntro} is defined only up to an overall monodromy-dependent constant $C(M)$, and extending it to a closed one-form on $T^*(\mathcal{M} \times \mathbb{T})$ allows to determine the asymptotic behaviour of the tau function near its critical points, while revealing the symplectic properties of the tau function.

In \cite{Its2016,10.1093/imrn/rnv375,Lisovyy2016}, building upon an earlier work \cite{Bertola2010,Bertola2021correction}, a procedure to construct such an extended closed one-form was presented based on the Riemann-Hilbert approach to Painlev\'e equations.
The one-form was explicitly written for the cases of Painlev\'e VI, II, III\(_1\), and I, where it was used to obtain the ratios of the corresponding tau functions at the critical points, known as the connection constants, a longstanding problem in the theory of Painlev\'e equations. A Hamiltonian approach to this construction was put forward in \cite{Its2018} for all Painlev\'e equations and the Schlesinger system. The closed one-form for the Schlesinger system was written explicitly in \cite{Bertola2019}, with an elegant interpretation of the tau function as the generating function of the monodromy symplectomorphism, {\it i.e.} a homomorphism from a symplectic leaf in the space of coefficients $\mathcal{A}$ of the system to a symplectic leaf in the monodromy manifold $\mathcal{M}$.

In this paper, we determine the closed one-form $\dd\log\T$ for tau functions on a torus with regular singularities by using their Fredholm determinant representation \cite{DelMonte2020}, and show that they generate the monodromy symplectomorphism.

A key technique in our construction involves the Fredholm determinant representation, which is obtained using a pants decomposition of the torus (see theorem 1, \cite{DelMonte2020}). Specifically, the pants decomposition of the torus with $n$ punctures consists of $n$ spheres with 3 simple poles, which translates to describing the local behaviour of the solution to the linear system we call $L$ in terms of the solution to an appropriate linear problem defined on a sphere with three simple poles denoted by $L_{3pt}$. It then turns out that the kernel of the Fredholm determinant is completely described by $n$ such three-point solutions with suitable shifts that capture the topology of the torus. The tau function of the torus is then related to the Fredholm determinant by a proportionality factor. The pith of the monodromy dependence of the tau function therefore lies in understanding the derivative of the Fredholm determinant w.r.t the monodromy data. The resulting one-form has the structure 
\begin{align}
    \dd\log\T=\omega-\omega_{3pt}, && \dd=\dd_t+\dd_{\cal M},\label{SPLIT}
\end{align}
where $\omega$ depends on the data coming from the global properties of $L$, whereas $\omega_{3pt}$ depends on the local behaviour described by $L_{3pt}$. Moreover, $\omega$ depends on the monodromy data and the times, while $\omega_{3pt}$ depends only on the monodromy data. Such a structure of the one-form is instrumental in obtaining the connection constant, which will be the subject of an upcoming paper. We also observe that our approach does not rely on the information of the asymptotics or additional assumptions to obtain the closedness. 

This paper is structured as follows: we setup the linear system on an $n$-punctured torus, describe the pants decomposition, and briefly recap the construction of the Fredholm determinant in section \ref{sec:setup}, we compute the monodromy dependence of the Fredholm determinant in proposition \ref{prop:derFred} and obtain the closed one-form in theorem \ref{thm:TauGenFn}, highlighting the splitting described in \eqref{SPLIT}. In section \ref{subsec:1pt}, we use the example of the torus with one puncture to illustrate the role of the one-form $\dd \log\mathcal{T}$ as the generating function of the monodromy symplectomorphism in theorem \ref{thm:GenFnCM}. 

Throughout this paper we use the following notation
\begin{notation}\label{not:bold}
 Given an N-tuple of parameters $(\xi_1,\dots\xi_N)$, and a function $g(\xi_i)$, $i=1,\dots,N$ of these parameters, we define
 \begin{equation}
     g(\bs\xi):=\diag\left(g(\xi_1),\dots,g(\xi_N)\right).
 \end{equation}
 In particular, when $g(\xi_i)=\xi_i$, this is
 \begin{equation}
     \bs\xi=\diag\left(\xi_1,\dots,\xi_N \right).
 \end{equation} 
\end{notation}

\section*{Acknowledgements}
We especially thank Oleg Lisovyy, for the suggestion of computing the one-form $\dd\log\T$ by using the Fredholm determinant, and H.D thanks Andrei Prokhorov for numerous discussions. Final stages of the draft were completed during F.D.M and H.D's residence at the Isaac Newton Institute during the Fall 2022 semester and they thank the M\o ller institute and organizers of the program "Applicable resurgent asymptotics: towards a universal theory" for hospitality during the final stages of the draft, supported by EPSRC grant no EP/R014604/1. H.D. acknowledges the support of Australian Research Council Discovery Project \#DP200100210, and her INI visit is supported by the Simons foundation. Part of this work was done during H.D's residence at the Mathematical Sciences Research Institute in Berkeley, California, during the Fall 2021 semester supported by the National Science Foundation Grant No. DMS-1928930.
P.G. would like to thank Max Planck Institute for Mathematics in Bonn, where a part of this work was done, for the hospitality during the complicated times, and Geneva university for hosting him now.
The work of P.G. is partially supported by the Fonds National Suisse and by the NCCR “The Mathematics of Physics” (SwissMAP).
P.G. would also like to thank the defenders of Ukraine, who are saving his Motherland.

\section{Isomonodromic deformations on the torus and tau function}\label{sec:setup}
In this section, we setup the $SL(N)$ linear system on the $n$-point torus, introduce the pants decomposition with corresponding 3-point local solutions, and briefly describe the construction of the Fredholm determinant representation of the isomonodromic tau-function.

 \subsection{Setup}
Isomonodromic deformations on a torus with $n$ simple poles can be characterised by the following system of linear differential equations \cite{Takasaki2003,Levin2013}
\begin{align}
  \begin{array}{c}
     \frac{\partial}{\partial z}\Phi\left(z\right) =\Phi\left(z\right) L_{z}\left(z\right),\\ \\
     \quad (2\pi i) \frac{\partial}{\partial\tau} \Phi\left(z\right) =  \Phi\left(z\right) L_\tau\left(z\right), \qquad
     \frac{\partial}{\partial z_k} \Phi\left(z\right)=\Phi\left(z\right)L_{k}\left(z\right),
\end{array}
\label{lax_general}
\end{align}
with $\Phi(z)\in SL(N,\mathbb{C})$. Here $z$ is the coordinate on the $n$-point torus $C_{1, n}$ viewed as the identification space $z\sim z+r+\ell\tau$, $r, \ell\in\mathbb{N}$, with singularities at the points $z_{k}$ for $k= 1\dots n$, $\tau \in \mathbb{H}$ is the modular parameter of the torus.
The Lax matrices $L_z,L_{\tau}, L_k\in\mathfrak{sl}_N$ have elements
\begin{equation}\label{eq:nptL}
    (L_z)_{ij}(z)  =\delta_{ij}\left\{P_i+\sum_{k=1}^n\frac{\theta_1'(z-z_k)}{\theta_1(z-z_k)}(A_k)_{ii} \right\}
    -(1-\delta_{ij})\sum_{k=1}^n x(Q_j-Q_i, z-z_k)
    (A_k)_{ij},
\end{equation}
\eqs{\label{eq:nptMk}
    (L_k)_{ij}(z)=-\delta_{ij}\frac{\theta_1'(z-z_k)}{\theta_1(z-z_k)}(A_k)_{ii}+ (1-\delta_{ij}) x(Q_j-Q_i, z-z_k)
    (A_k)_{ij},
    }
\eq{\label{eq:nptMt}(L_\tau)_{ij}(z)=-\frac{1}{2}\delta_{ij}\sum_{k=1}^n\frac{\theta_1''(z-z_k)}{\theta_1(z-z_k)}(A_k)_{ii} -\sum_{k=1}^ny(Q_j-Q_i,z-z_k)(A_k)_{ij},
}
where $x, y$ are the Lam\'e functions defined as 
\begin{align}
    x(\xi, z) = \frac{\theta_{1}(z- \xi) \theta_{1}'(0)}{\theta_{1}(z) \theta_{1}(\xi)}, && y (\xi, z) = \partial_{\xi} x(\xi, z),
\end{align}
where $\theta_1$ is the Jacobi theta function 
\begin{equation}
    \theta_1(z):=\sum_{n\in\mathbb{Z}}(-1)^{n-\frac{1}{2}}e^{i\pi\tau\left(n+\frac{1}{2} \right)^2}e^{2\pi i\left(n+\frac{1}{2} \right)z}. \label{eq:Theta1def}
\end{equation}
The above linear system has two main properties.
\begin{itemize}
\item[1.] The matrices $A_{k}$ are diagonalizable
\begin{gather}\label{eq:AGmGinv}
    A_{k} = G_{k}^{-1} {\bs m_{k}} G_{k}
\end{gather}
and satisfy the constraint
\begin{equation}\label{cons:CONSTRAINT}
    \sum_{k=1}^n (A_k)_{ii}=\sum_{k=1}^n(G_k^{-1}\bs m_kG_k)_{ii} =0.
\end{equation}
Moreover, 
\begin{align}\label{cons:con2}
    m_{i}-m_{j}\notin \mathbb{Z}, && \sum_{j=1}^NP_j=\sum_{j=1}^NQ_j=0.
\end{align}
\item[2.] The matrices $L_z,\,L_k,\,L_{\tau}$ have the following transformation (using notation \ref{not:bold}) under the shift $z\mapsto z+\tau$
\begin{gather}
    L_z(z+\tau)=e^{-2\pi i\bs Q}L(z)e^{2\pi i\bs Q}, \nonumber \\
    L_k(z+\tau)=e^{-2\pi i\bs Q}L_k(z)e^{2\pi i\bs Q}+2\pi i \,\diag\left((A_k)_{11},\dots,(A_k)_{NN} \right), \label{eq:LMMTransform} \\
        L_\tau(z+\tau)=e^{-2\pi i\bs Q}\left(L_\tau(z)+L(z) \right)e^{2\pi\bs Q}-2\pi i\bs P \nonumber.
\end{gather}
\end{itemize}
The solution of the linear system \eqref{lax_general}, under the $z\mapsto z+ \tau$ shift therefore transforms as
\begin{equation}
   \Phi(z+\tau)=M_B\Phi(z)e^{2\pi i\bs Q},\label{eq:zdep_MB}
\end{equation}
where $M_{B} \in SL(N)$ is the B-cycle monodromy, and the monodromies around the punctures $z_k$ and the A-cycle monodromy are respectively
\begin{align}
    M_{k} = C_{k}e^{2\pi i {\bs m}_{k}} C_{k}^{-1}, && M_{A} =S_1 e^{2\pi i {\bs a}_{1}}S_1^{-1}, \label{eq:ntorus_mon}
\end{align}
with the constraint
\begin{gather}
    M_{B}^{-1} M_{A}^{-1} M_{B} M_{A}\prod_{k=1}^{n} M_{k} = \mathbb{1}.  \label{eq:n_moncons}
\end{gather}
\begin{definition}
\begin{enumerate}
    \item The Lax matrices $L_z(z)$ in equation \eqref{eq:nptL}  are described by the following space: {\small
\begin{equation}\label{eq:A1n}
    \mathcal{A}_{1,n}:=\bigg\{\tau,\,(G_k,\bs m_k,z_k)_{k=1}^n,\,( P_j, Q_j)_{j=1}^N:\tau\in\mathbb{H},\,z_k\in T^2_\tau,\,G_k\in SL(N),
    \,P_j,Q_j\in\mathbb{C},\,\ \eqref{cons:CONSTRAINT}, \eqref{cons:con2} \bigg\}/\sim,
\end{equation}}
where $\sim$ is the equivalence relation  $G_k\rightarrow G_k D$, where $\,D\in SL(N)$ is diagonal\footnote{The full gauge group acts as $G_k\rightarrow G_k H$, where $H\in SL(N)$ is an arbitrary matrix. However, the choice of diagonal $e^{2\pi i\bs Q}$ uniquely fixes $H$ to be diagonal}. The dimension of this space is
\begin{equation}
    \dim\mathcal{A}_{1,n}= n(N^2-1)+n(N-1) +n+1.
\end{equation}
$\mathcal{A}_{1,n}$ can be viewed as the symplectic reduction of the moduli space of flat $SL(
N,\mathbb{C})$ connections on the $n$-punctured torus \cite{Alekseev1995,hitchin1997frobenius,Korotkin1995,Levin1999,Krichever2001,Levin2013}).
\item The extended character variety of $SL(N)$ flat connections on $T_\tau^2\setminus\{z_1,\dots,z_n\}$(see \cite{Bertola2022} for genus zero case) is
 \begin{align}\label{eq:M1n}
    \mathcal{M}_{1,n} = \left\lbrace M_{A}, M_{B}, (C_k,\bs m_k)_{k=1}^n:\, M_A,M_B,C_k\in SL(N),\, \eqref{eq:ntorus_mon}, \eqref{eq:n_moncons}\right\rbrace/\sim ,
\end{align}
where $/\sim$ means that we identify monodromy representations related by an overall conjugation.
The dimension of the extended character variety is 
\begin{equation}
\dim\mathcal{M}_{1,n}= n(N^2-1)+n(N-1).
\end{equation}
For arbitrary genus, $\dim\mathcal{M}_{g,n}=(N^2-1)(2g-2+n)+n(N-1)$.
\end{enumerate}
\end{definition}

The usual (non-extended) character variety would be \begin{align}
    \mathcal{M}_{1,n}^{(0)} = \left\lbrace M_{A}, M_{B},M_1,\dots,M_n \in SL(N) \vert \eqref{eq:n_moncons}\right\rbrace/\sim, && && \dim\mathcal{M}_{1,n}^{(0)}=n(N^2-1),
\end{align}
and the standard space of coefficients
\begin{equation}\label{eq:A1n}
    \mathcal{A}_{1,n}^{(0)}:=\bigg\{\tau,\,(A_k,z_k)_{k=1}^n,\,( P_j, Q_j)_{j=1}^N:\tau\in\mathbb{H},\,z_k\in T^2_\tau,\,A_k\in SL(N),
    \,P_j,Q_j\in\mathbb{C},\,\ \eqref{cons:CONSTRAINT}, \eqref{cons:con2} \bigg\}/\sim,
\end{equation}
where now $A_k\sim  D^{-1}A_k D$, and the dimension 
\begin{equation}
    \dim\mathcal{A}_{1,n}^{(0)}=n(N^2-1)+n+1.
\end{equation}

By using gauge freedom, it is always possible to choose $S_1=\mathbb{1}$ in the A-cycle monodromy \eqref{eq:ntorus_mon}, i.e.
\begin{equation}\label{new_mon_a}
    M_A\equiv e^{2\pi i\bs a_1}.
\end{equation}
The local behaviour of the solution to the linear system in a tubular neighbourhood of the puncture $z_k$ is
\begin{equation}\label{eq:LocalTorus}
    \Phi \left( z \rightarrow z_{k} \right) = C_{k} \left(z- z_{k}\right)^{\bs m_k} \left( \mathbb{1} + \sum_{l=1}^{\infty} g_{k,l} (z-z_{k})^{l} \right) G_{k},
\end{equation}
where the matrices $C_k, G_k$ diagonalize the monodromies \eqref{eq:ntorus_mon} and the residue matrices \eqref{eq:AGmGinv} respectively. There is an ambiguity of the form
\begin{align}\label{eq:DK}
    C_k\mapsto C_kD_k^{-1}, && g_{k,l}\mapsto D_kg_{k,l}D_k^{-1}, && G_k\mapsto D_k G_k,
\end{align}
with $D_k$ diagonal, that does not change the asymptotics \eqref{eq:LocalTorus}, and it amounts to a change of normalization for the eigenvectors of $M_k,\, A_k$. The extended spaces \eqref{eq:M1n}, \eqref{eq:A1n} differ from the non-extended ones by the inclusion of the parameters that are changed by the transformation \eqref{eq:DK}. These parameters turn out to be canonically conjugated to \(\boldsymbol{m}_k\) \cite{Bertola2019}.

The matrices $g_{k,l}$ are computed recursively with the $i,j$ component given by
\begin{equation}\label{eq:g1_torus}
\begin{split}
    \left[G_{k}^{-1}\left(g_{k,1} + \left[ \bs m_{k}, g_{k,1} \right] \right) G_{k} \right]_{ij}& = \delta_{ij}\left\{P_i+\sum_{k'\neq k}^n\frac{\theta_1'(z_k-z_{k'})}{\theta_1(z_{k}-z_{k'})}(A_{k'})_{ii} \right\}\\
    & -(1-\delta_{ij})\sum_{k=1}^n x(Q_i-Q_j, z_k-z_{k'})
    (A_{k'})_{ij}.
\end{split}
\end{equation}

The isomonodromic time evolution\footnote{Here by isomonodromy we mean that the $C_k$'s are constant, as opposed to just the monodromies. This is the correct notion for the extended spaces $\cA_{1,n}$ and $\cM_{1,n}$.} arising from the compatibility of \eqref{lax_general} is generated by the $n + 1$ Poisson commuting Hamiltonians
\begin{align}\label{eq:IsomHams}
	H_{z_k}:=\res_{z=z_k}\frac{1}{2}\tr L(z)^2, && H_\tau:=\oint_Adz\frac{1}{2}\tr L(z)^2.
\end{align}
\begin{definition}
The isomonodromic tau function $\T_H$ is then defined as
\begin{align}\label{eq:TauDef}
	\partial_{z_k}\log\T_H=H_k, && 2\pi i\partial_\tau\log\T_H=H_\tau.
\end{align}	
\end{definition}

\subsection{Pants decomposition of the $n$-point torus and Hilbert spaces}\label{sec:pants}
The $n$-punctured torus can be decomposed into $n$ trinions, that we choose to be glued along copies of the A-cycle as in Figure \ref{fig:nptTorusPants}. There is a three-point problem associated to each trinion $\mathscr{T}^{[k]}$ 
\begin{gather}
    \partial_z \Phi_{3pt}^{[k]}(z)=\Phi_{3pt}^{[k]}(z) L_{3pt}^{[k]}(z), \nonumber\\ L_{3pt}^{[k]}(z)=-2\pi i A_{-}^{[k]}-2\pi i\frac{A_0^{[k]}}{1-e^{2\pi iz}},\label{eq:3ptk}
\end{gather}
with diagonalizable residue matrices
\begin{align}\label{eq:AGmGinv3pt}
    A_{-}^{[k]}=(G_{-}^{[k]})^{-1} \bs a_k G_{-}^{[k]}, && A_0^{[k]}=(G_0^{[k]})^{-1}\bs m_kG_0^{[k]},
\end{align} 
\begin{align}
        A_{+}^{[k]} = - A_{-}^{[k]}- A_{0}^{[k]} = (G_{+}^{[k]})^{-1} \bs a_{k+1} G_{+}^{[k]},
\end{align}
for $k =1, \dots n$.
The local solution on each trinion \(\Phi_{3pt}^{[k]}(z)\) is such that the ratio
\begin{equation*}
\Phi^{[k]}_{3pt}(z-z_k)^{-1}\Phi(z)
\end{equation*}
is regular and single-valued around \(z=z_k\), with $\Phi^{[k]}_{3pt}(z-z_{k})$ approximating the analytic behavior of \(\Phi(z)\) in the trinion $\mathscr{T}^{[k]}$.

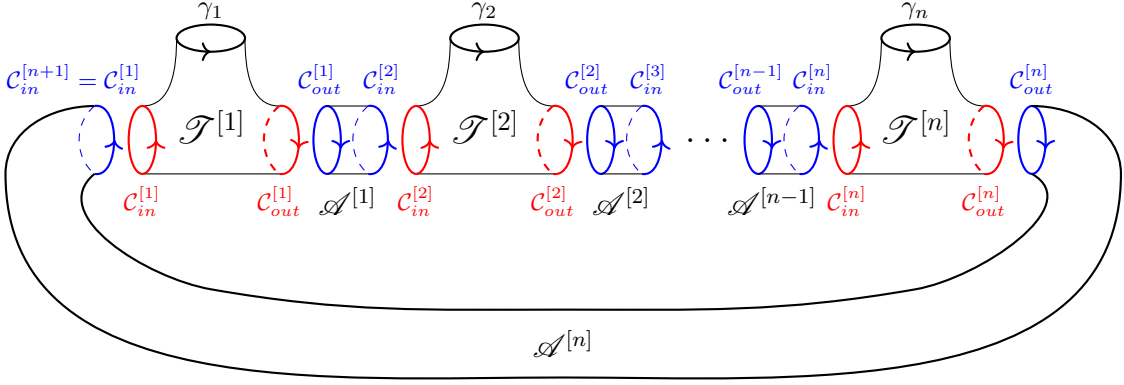
\begin{figure}[H]
\begin{center}
\begin{tikzpicture}[scale=0.9]

\draw[thick,decoration={markings, mark=at position 0.75 with {\arrow{>}}}, postaction={decorate}] (-5,1.5) circle[x radius=0.5, y radius=0.2];
\draw(-6,0.5) to[out=0,in=270] (-5.5,1.5);
\draw(-4,0.5) to[out=180,in=270] (-4.5,1.5);
\draw(-6,-0.5) to[out=0,in=180] (-4,-0.5);
\draw[thick, red,decoration={markings, mark=at position 0.6 with {\arrow{>}}}, postaction={decorate}]  (-4,0.5) to[out=10,in=-10] (-4,-0.5);
\draw[thick,dashed, red] (-4,0.5) to[out=220,in=-220] (-4,-0.5);
\draw[thick,red,decoration={markings, mark=at position 0 with {\arrow{>}}}, postaction={decorate}] (-6,0) circle[y radius=0.5, x radius=0.2];

\draw[thick,blue,decoration={markings, mark=at position 0 with {\arrow{<}}}, postaction={decorate}] (-3.3,0) circle[y radius=0.5, x radius =0.2];
\draw (-3.3,-0.5) to (-2.7,-0.5);
\draw (-3.3,0.5) to (-2.7,0.5);
\draw[thick, blue,decoration={markings, mark=at position 0.6 with {\arrow{<}}}, postaction={decorate}]  (-2.7,0.5) to[out=10,in=-10] (-2.7,-0.5);
\draw[dashed, blue] (-2.7,0.5) to[out=220,in=-220] (-2.7,-0.5);

\draw[thick,decoration={markings, mark=at position 0.75 with {\arrow{>}}}, postaction={decorate}] (-1,1.5) circle[x radius=0.5, y radius=0.2];
\draw(-2,0.5) to[out=0,in=270] (-1.5,1.5);
\draw(0,0.5) to[out=180,in=270] (-0.5,1.5);
\draw(-2,-0.5) to[out=0,in=180] (0,-0.5);
\draw[thick,red,decoration={markings, mark=at position 0 with {\arrow{>}}}, postaction={decorate}] (-2,0) circle[y radius=0.5, x radius=0.2];
\draw[thick, red,decoration={markings, mark=at position 0.6 with {\arrow{>}}}, postaction={decorate}]  (0,0.5) to[out=10,in=-10] (0,-0.5);
\draw[thick,dashed, red] (0,0.5) to[out=220,in=-220] (0,-0.5);

\draw[thick,blue,decoration={markings, mark=at position 0 with {\arrow{<}}}, postaction={decorate}] (0.7,0) circle[y radius=0.5, x radius =0.2];
\draw (0.7,-0.5) to (1.3,-0.5);
\draw (0.7,0.5) to (1.3,0.5);
\draw[thick, blue,decoration={markings, mark=at position 0.6 with {\arrow{<}}}, postaction={decorate}]  (1.3,0.5) to[out=10,in=-10] (1.3,-0.5);
\draw[dashed, blue] (1.3,0.5) to[out=220,in=-220] (1.3,-0.5);

\node at (2.3,0) {{\Large\dots}};

\draw[thick,blue,decoration={markings, mark=at position 0 with {\arrow{<}}}, postaction={decorate}] (3,0) circle[y radius=0.5, x radius =0.2];
\draw (3,-0.5) to (3.6,-0.5);
\draw (3,0.5) to (3.6,0.5);
\draw[thick, blue,decoration={markings, mark=at position 0.6 with {\arrow{<}}}, postaction={decorate}]  (3.6,0.5) to[out=10,in=-10] (3.6,-0.5);
\draw[dashed, blue] (3.6,0.5) to[out=220,in=-220] (3.6,-0.5);

\draw[thick,decoration={markings, mark=at position 0.75 with {\arrow{>}}}, postaction={decorate}] (5.3,1.5) circle[x radius=0.5, y radius=0.2];
\draw(4.3,0.5) to[out=0,in=270] (4.8,1.5);
\draw(6.3,0.5) to[out=180,in=270] (5.8,1.5);
\draw(4.3,-0.5) to[out=0,in=180] (6.3,-0.5);
\draw[thick,red,decoration={markings, mark=at position 0 with {\arrow{>}}}, postaction={decorate}] (4.3,0) circle[y radius=0.5, x radius=0.2];
\draw[thick, red,decoration={markings, mark=at position 0.6 with {\arrow{>}}}, postaction={decorate}]  (6.3,0.5) to[out=10,in=-10] (6.3,-0.5);
\draw[thick,dashed, red] (6.3,0.5) to[out=220,in=-220] (6.3,-0.5);

\draw[thick, blue,decoration={markings, mark=at position 0.6 with {\arrow{<}}}, postaction={decorate}]  (-6.7,0.5) to[out=10,in=-10] (-6.7,-0.5);
\draw[dashed, blue] (-6.7,0.5) to[out=220,in=-220] (-6.7,-0.5);

\draw[thick,blue,decoration={markings, mark=at position 0 with {\arrow{<}}}, postaction={decorate}] (7,0) circle[y radius=0.5, x radius =0.2];

\draw[thick] (-6.7,0.5) to[out=180,in=90] (-8,-0.5) to[out=270,in=180] (0.15,-3.5) to[out=0,in=270] (8.3,-0.5) to[out=90,in=0] (7,0.5);

\draw[thick] (-6.7,-0.5) to[out=220,in=170] (-5,-2.2) to[out=-10,in=180] (0.15,-2.5) to[out=0,in=190] (5.3,-2.2) to[out=10,in=-30] (7,-0.5);


\node at (-5,0.2) {{\Large$\mathscr{T}^{[1]}$}};
\node at (-1,0.2) {{\Large$\mathscr{T}^{[2]}$}};
\node at (5.3,0.2) {{\Large$\mathscr{T}^{[n]}$}};

\node at (-6,-0.9) {{\footnotesize${\color{red} {\mathcal{C}}_{in}^{[1]}}$}};
\node at (-5,1.9) {{\small$ \gamma_{1}$}};
\node at (-4,-0.9) {{\footnotesize${\color{red} {\mathcal{C}}_{out}^{[1]}}$}};

\node at (-3.4,0.9) {{\footnotesize${\color{blue} {\mathcal{C}}_{out}^{[1]}}$}};
\node at (-2.5,0.9) {{\footnotesize${\color{blue} {\mathcal{C}}_{in}^{[2]}}$}};

\node at (-2,-0.9) {{\footnotesize${\color{red} {\mathcal{C}}_{in}^{[2]}}$}};
\node at (-1,1.9) {{\small$ \gamma_{2}$}};
\node at (0,-0.9) {{\footnotesize${\color{red} {\mathcal{C}}_{out}^{[2]}}$}};

\node at (0.5,0.9) {{\footnotesize${\color{blue} {\mathcal{C}}_{out}^{[2]}}$}};
\node at (1.4,0.9) {{\footnotesize${\color{blue} {\mathcal{C}}_{in}^{[3]}}$}};

\node at (2.9,0.9) {{\footnotesize${\color{blue} {\mathcal{C}}_{out}^{[n-1]}}$}};
\node at (3.8,0.9) {{\footnotesize${\color{blue} {\mathcal{C}}_{in}^{[n]}}$}};

\node at (4.3,-0.9) {{\footnotesize${\color{red} {\mathcal{C}}_{in}^{[n]}}$}};
\node at (5.3,1.9) {{\small$ \gamma_{n}$}};
\node at (6.3,-0.9) {\footnotesize{${\color{red} {\mathcal{C}}_{out}^{[n]}}$}};

\node at (7,0.9) {{\footnotesize${\color{blue} {\mathcal{C}}_{out}^{[n]}}$}};
\node at (-7,0.9) {{\footnotesize${\color{blue} {\mathcal{C}}_{in}^{[n+1]}={\mathcal{C}}_{in}^{[1]} }$}};

\node at (-3,-0.9) {{\large$\mathscr{A}^{[1]}$}};
\node at (1,-0.9) {{\large$\mathscr{A}^{[2]}$}};
\node at (3.2,-0.9) {{\large$\mathscr{A}^{[n-1]}$}};
\node at (0.15,-3) {{\large$\mathscr{A}^{[n]}$}};

\end{tikzpicture}
\end{center}
\caption{Pants decomposition for the $n$-punctured torus}\label{fig:nptTorusPants}
\end{figure}

In terms of the contours defined above, the fundamental domain for the torus is the parallelogram in the figure \ref{fig:TorusParalgram_11}.
 \begin{figure}[H]
\centering
\begin{tikzpicture}[scale = 4]
\draw[ultra thick, blue,decoration={markings, mark=at position 0.5 with {\arrow{<}}}, postaction={decorate}] (0,0) to (1,0);
\draw[ultra thick, red,decoration={markings, mark=at position 0.5 with {\arrow{>}}}, postaction={decorate}] (0.72,0.5) to (1.72,0.5);
\draw[thick] (0,0) to (0.72,0.5);
\draw[thick] (1,0) to (1.72,0.5);

\node at (0.5,-0.1) {{\color{blue}$\cin^{[1]}$}};
\node at (1.22,0.6) {{\color{red}$\cout^{[n]}$}};

\node at (-0,-0.07) {$-\frac{\tau+1}{2}$};
\node at (1,-0.07) {$-\frac{\tau-1}{2}$};
\node at (0.72,0.6) {$\frac{\tau-1}{2}$};
\node at (1.72,0.6) {$\frac{\tau+1}{2}$};
\end{tikzpicture}
\caption{The fundamental domain of the torus}
\label{fig:TorusParalgram_11}
\end{figure}
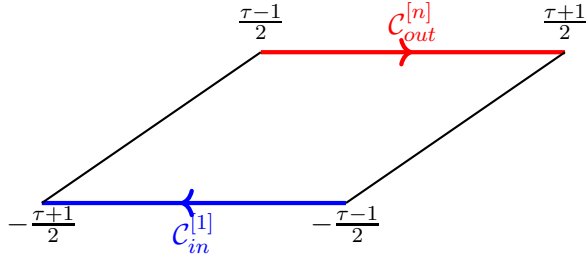

The parameters $\bs m_k$, $\bs a_k$ are local monodromy exponents, i.e the monodromies around the contours $\gamma_{k}$, $\cC_{in}^{[k]}$, $\cC_{out}^{[k-1]}$,  are respectively
{\begin{align}
    M_k=C_k e^{2\pi i\,\bs m_k}C_k^{-1}, && M_{\mathcal{C}_{in}^{[k]}}=S_k e^{2\pi i\, \bs a_k}S_k^{-1}, && M_{\mathcal{C}_{out}^{[k]} } = S_{k+1} e^{-2\pi i\, \bs a_{k+1}}S_{k+1}^{-1},
    \label{eq:1nmonodromy}
\end{align}}
for $k= 1,\dots n $, where the $C_{k}$, $S_{k}$ are constant matrices with $S_k$'s assuming the form
\begin{align}\label{id:S_1n}
    S_{1} = \mathbb{1}, && S_{n+1} = M_{B}^{-1}, && \bs a_{n+1} = \bs a_{1}.
\end{align}
Note that $\bs a_{1}$ corresponds to the usual $A$-cycle as described in \eqref{new_mon_a}. 
The pants decomposition induces a homomorphism of monodromy groups $  \pi_{1}\left(C_{0,n+2}\right)\rightarrow\pi_{1}\left(C_{1,n}\right)$ as can be seen from the the constraint \eqref{eq:n_moncons}
\begin{gather}
     M_{A}\prod_{k=1}^{n} M_{k} M_{B}^{-1} M_{A}^{-1} M_{B} = \mathbb{1} = M_{\mathcal{C}_{in}^{[1]}} \left(\prod_{k=1}^n M_{k} \right)  M_{\mathcal{C}_{out}^{[n]}}.
\end{gather}
The local behaviour of the three-point solution $\Phi^{[k]}$ around the point $z_k$ is then
\begin{equation}
    \Phi^{[k]}_{3pt} \left( z \rightarrow 0 \right) = C_{k} z^{\bs m_k} \left( \mathbb{1} + \sum_{l=1}^{\infty} g_{1,l}^{[k]} z^{l} \right) G_{0}^{[k]}, \label{eq:Local3pt}
\end{equation}
while its local behavior on the circles $\cC_{in}^{[k]}$ and $\cC_{out}^{[k]}$ is
\begin{equation}
\Phi_{3pt}^{[k]}\big\vert_{z\in \cC_{in}^{[k]}}=S_k e^{2\pi iz\bs a_k}\left(\mathbb{1}+\sum_{l=1}^{\infty}g_{-,l}^{[k]}e^{-2\pi i l z} \right)G_{-}^{[k]},\label{eq:3ptpinfty}
\end{equation}
\begin{equation}
\Phi_{3pt}^{[k]}\big\vert_{z\in \cC_{out}^{[k]}}=S_{k+1}e^{2\pi iz\bs a_{k+1}}\left(\mathbb{1}+\sum_{l=1}^{\infty}g_{+,l}^{[k]}e^{2\pi i l z} \right)G_{+}^{[k]}.\label{eq:3ptminfty}
\end{equation}
Once again note that $C_{k}, S_{k}, G_{0}^{[k]}, G_{-}^{[k]}, G_{+}^{[k]}$ are the diagonlization matrices in \eqref{eq:AGmGinv3pt}, \eqref{eq:1nmonodromy}. As before, the matrices $g_{1,l}$ can be computed recursively with the $1,1$ matrix entry given by
\begin{gather}\label{eq:g1_3pt}
    g_{1,1}^{[k]} + \left[ \bs m_{k}, g_{1,1}^{[k]} \right] = -2\pi i G_{0}^{[k]} A_{0}^{[k]} \left( G_{0}^{[k]} \right)^{-1}.
\end{gather}

We now associate Hilbert spaces to each of the boundary contours $\cC_{in}^{[k]}, \cC_{out}^{[k]}$ specified in Figure \ref{fig:nptTorusPants}. The projection operators on these spaces provide the building blocks to write the tau-function as a Fredholm determinant.
The total Hilbert space $\cH$ is decomposed into a direct sum of spaces $\cH^{[k]}$ corresponding to each pair of pants:
\eqs{\cH := \bigoplus_{k=1}^{n} \cH^{[k]} = \cH_{+} \oplus \cH_{-}, \label{eq:total_Hilbert}}
where
\eqs{\cH_{\pm}:=\bigoplus_{k=1}^{n} \left( \cH^{[k]}_{in, \mp} \oplus \cH^{[k]}_{out, \pm} \right).\label{eq:Hilbert_decomposition}}

\begin{definition}\label{def:YcalYnpoint}
We associate the single-valued matrix-valued functions $\Psi(z)$, $\Psi^{[k]}_{3pt}(z)$ defined on the boundary circles of the pants decomposition, to the functions $\Phi(z)$, $\Phi^{[k]}_{3pt}(z)$ in \eqref{lax_general}, \eqref{eq:3ptk} respectively\footnote{The $z-\tau$ in $\Psi$ is so that $\Psi_{out,n}$ is "identified" with $\Psi_{in,1}$, in the sense that the only difference is the twist}

\begin{align} \label{psi_body}
\Psi(z)\vert_{\cC^{[k]}_{out}}:=e^{-2\pi i(z-\delta_{k,n}\tau)\bs a_{k+1}}S_{k+1}^{-1}\Phi(z)\vert_{\cC^{[k]}_{out}}
, && \Psi(z)\vert_{\cC^{[k]}_{in}}:=e^{-2\pi iz\bs a_{k}}S_{k}^{-1}\Phi(z)\vert_{\cC^{[k]}_{in}}
, &&k=1,\dots,n
\end{align}
\begin{align} \label{3pt_body}
\Psi^{[k]}_{3pt}(z)\vert_{\cC^{[k]}_{out}}:=e^{-2\pi i{(z-\delta_{k,n}\tau)}\bs a_{k+1}}S_{k+1}^{-1}\Phi^{[k]}_{3pt}(z-z_k)\vert_{\cC^{[k]}_{out}}
, && \Psi^{[k]}_{3pt}(z)\vert_{\cC^{[k]}_{in}}:=e^{-2\pi i z\bs a_{k}}S_{k}^{-1}\Phi^{[k]}_{3pt}(z-z_k)\vert_{\cC^{[k]}_{in}}
,
\end{align}
where $\delta_{k,n}$ is the Kr\"onecker delta and the identities \eqref{id:S_1n} hold.
\end{definition}
\subsection{Fredholm determinant representation of the tau-function}\label{sec:tau}
We define the projection operators $\cP_{\Sigma}$, $\cP_{\oplus}$ on the Hilbert space defined in \eqref{eq:total_Hilbert} in terms of the solutions to the $n$-point linear sytem \eqref{psi_body}, and the solutions \eqref{3pt_body} to the three-point problems respectively\footnote{The Cauchy kernel are written in the cylindrical coordinates below, as they are more natural for our parametrization of the torus.}:
\begin{itemize}
    \item The operator $\cP_{\Sigma}$ is defined as
\eqs{\label{eq:PSigman}
    \left(\cP_{\Sigma}f\right)(z) := \oint_{\cC_{\Sigma}} \frac{\dd w}{2\pi i} \Psi(z) \Xi_{N}(z,w) \Psi (w){^{-1}} f(w), }
where 
\eqs{\cC_{\Sigma} := \bigcup_{k=1}^{n} \cout^{[k]} \cup \cin^{[k+1]}, && \cin^{[n+1]} := \cin^{[1]},}
and the (twisted) Cauchy kernel on the $n$-point torus
\begin{equation}\label{eq:Xi_N_exp}
\begin{split}
\Xi_{N}(z,w) & =\diag\left( \frac{\theta_1(z-w+Q_1-\rho)\theta_1'(0)}{\theta_1(z-w)\theta_1(Q_1-\rho)},\dots,\frac{\theta_1(z-w+Q_N-\rho)\theta_1'(0)}{\theta_1(z-w)\theta_1(Q_N-\rho)}\right)\\
 & = \frac{\theta_1(z-w+\bs Q-\rho)\theta_1'(0)}{\theta_1(z-w)\theta_1(\bs Q-\rho)}
,
\end{split}
\end{equation}
has the following transformation
\begin{align}\label{eq:TauShiftXiN}
    \Xi_N(z+\tau,w)=e^{2\pi i(\bs Q-\rho)}\, \Xi_N(z,w), && \Xi_N(x,w+\tau)=\Xi_N(z,w) e^{-2\pi i (\bs Q-\rho)}.
\end{align}

\item The operator $\cP^{[k]}$ is defined as
\begin{equation}\label{pk}
    \left(\cP_{\oplus}^{[k]} f^{[k]}\right)(z)  := \int_{\cin^{[k]} \cup\cout^{[k]}} \dd w \frac{\Yt^{[k]}(z) \Yt^{[k]} (w){^{-1}}}{1-e^{-2\pi i(z-w)}} f^{[k]}(w),
\end{equation}
and
\begin{align}
\cP_{\oplus} := \sum_{k=1}^{n} \cP_{\oplus}^{[k]}.
\end{align}
\end{itemize}
In the paper \cite{DelMonte2020} we proved that the tau function  \eqref{eq:TauDef} has the following Fredholm determinant representation:
    \begin{equation}
    \begin{split}\T_{H}(\tau) &= \det_{\cH_{+}}\left[\cP_{\Sigma, +}^{-1} \cP_{\oplus,+} \right] e^{i\pi\tau\tr\left(\bs a_1^2+\frac{\mathbb{1}}{6} \right)} e^{-i \pi N \rho} \prod_{i=0}^N\frac{\eta(\tau)}{\theta_1\left(Q_i-\rho \right)}\prod_{k=1}^n e^{-i\pi z_k\left(\tr\bs a_{k+1}^2-\tr\bs a_{k}^2\right)}, \end{split}\label{eq:Thm2}
    \end{equation}    
where $\cP_{\Sigma, +}:=\cP_{\Sigma}\big|_{\cH_+}, \cP_{\oplus,+}\big|_{\cH_+}$, $Q_{i} \equiv Q_{i}(\tau, z_{1},...,z_{n})$ are the dynamical variables of the $SL(N)$ linear system \eqref{eq:nptL}, $ \bs a_k$ are the monodromy exponents defined in \eqref{eq:1nmonodromy}, and $\rho$ is an arbitrary parameter. Furthermore, the time derivative of the Fredholm determinant recovers the JMU one form:
\begin{gather}\label{eq:tder}
    \omega_{JMU}:=\left(\dd_{\tau}+ \sum_{k=1}^{n} \dd_{z_{k}} \right) \log\mathcal{T}_{H} = \frac{1}{2\pi i}H_\tau \dd\tau+\sum_{k=1}^{n}H_k \dd z_k.
\end{gather}
We will do this by explicitly computing the derivative of the tau function \eqref{eq:Thm2} with respect to the monodromy data.
\section{Monodromy dependence of the torus tau function}\label{sec:MainThms}
We begin with the following identity for the derivative of the Fredholm determinant \cite{Gavrylenko2016b}
\begin{align} \label{eq:Tr}
    \dd \log \det_{\mathcal{H}_{+}} \left[ \cP_{\Sigma,+}^{-1} \cP_{\oplus,+} \right] = - \tr_{\cH} \cP_{\oplus} \dd \cP_{\Sigma}, && \dd = \dd_{\mathcal{M}} + \dd_{\tau} + \sum_{k=1}^n \dd_{z_k}, && \dd_{\bullet} = \dd\bullet \partial_{\bullet},
 \end{align}
 where $\dd_{\cal M}$ is the total exterior derivative on the extended character variety $\mathcal{M}_{1,n}$.
\begin{proposition}\label{prop:derFred}
The derivative of the Fredholm determinant w.r.t the monodromy data is 
\begin{equation}\label{eq:der_Fred_mon} 
\begin{split}
&\dd_{\mathcal{M}}\log \det_{\mathcal{H}_{+}} \left[ \cP_{\Sigma,+}^{-1} \cP_{\oplus,+} \right]  = \tr\left({\bs P}\dd_{\mathcal{M}}{\bs Q}\right)+\dd_{\mathcal{M}}\log\left(\prod_{i=1}^N\theta_1(Q_i-\rho) \right) +\sum_{k=1}^n\tr\bs m_k\dd_{\mathcal{M}} G_{k}G_{k}^{-1}\\
&+\sum_{k=1}^n \left(\tr\bs a_k\dd_{\mathcal{M}} G_{-}^{[k]}\left(G_{-}^{[k]}\right)^{-1} - \tr\bs m_k\dd_{\mathcal{M}} G_0^{[k]}\left(G_0^{[k]}\right)^{-1}-\tr\bs a_{k+1}\dd_{\mathcal{M}} G_{+}^{[k]}\left(G_{+}^{[k]}\right)^{-1}\right) \\
&-i\pi \tr \tau \dd_{\mathcal{M}} {\bs a}_{1}^2 + \sum_{k=1}^{n} i \pi \tr z_{k} \dd_{\mathcal{M}} \left({\bs a}_{k+1}^2 - {\bs a}_{k}^2 \right),
\end{split}
\end{equation}
where (with the notation \ref{not:bold}) ${\bs m}_{k}$, ${\bs a}_{k}$ are the local monodromy exponents \eqref{eq:1nmonodromy}, ${\bs P, \bs Q}$ are the dynamical variables in the Lax matrix \eqref{eq:nptL}, $\rho$ is an arbitrary parameter, $G_{k},\, G_{\pm}^{[k]}$ are the eigenvector matrices \eqref{eq:AGmGinv}, \eqref{eq:AGmGinv3pt}.
\end{proposition}

\begin{proof}
We start from the following equality (see eq. 3.61 in \cite{DelMonte2020})
\begin{align}
-\tr_{\cH}\left[\cP_{\oplus}\dd_{\mathcal{M}}\cP_{\Sigma} \right]& =-\sum_{k=1}^{n} \oint_{\cC_{in}^{[k]}\cup\cC_{out}^{[k]}}\dd w \oint_{\cC_{in}^{[k]}\cup\cC_{out}^{[k]}}\frac{\dd z}{2\pi i} \frac{1}{1-e^{-2\pi i(z-w)}}\tr\bigg\{\Yt^{[k]}(z)\Yt^{[k]}(w)^{-1} \nonumber \\
& \hspace{7cm}\times \dd_{\mathcal{M}} \left( \Psi(w)\Xi_N(w,z)\Psi(z)^{-1} \right) \bigg\} \nonumber \\
& =- \sum_{k=1}^{n}\oint_{\cC_{in}^{[k]}\cup\cC_{out}^{[k]}}\frac{\dd z}{2\pi i} \tr\left\{ d_{\mathcal{M}} \Psi(z) \Psi(z)^{-1}\left(\partial_{z}\Yt^{[k]}(z) \left(\Yt^{[k]}(z)\right)^{-1} \right) \right\} \nonumber\\
&+ \sum_{k=1}^{n}\oint_{\cC_{in}^{[k]}\cup\cC_{out}^{[k]}}\frac{\dd z}{2\pi i} \tr\left\{ \dd_{\mathcal{M}} \Psi(z) \Psi(z)^{-1}\left(\partial_{z}\Psi(z) \Psi(z)^{-1}\right) \right\} \nonumber\\
& +\sum_{k=1}^{n} \oint_{\cC_{in}^{[k]}\cup\cC_{out}^{[k]}}\frac{\dd z}{2\pi i} \tr\left\{ \left[\frac{\theta_1'(\bs Q-\rho)}{\theta_1(\bs Q-\rho)}-i\pi\mathbb{1}_N \right]\Psi(z)^{-1}\dd_{\mathcal{M}} \Psi(z) \right\}  \nonumber \\
&=: I_{1} + I_{2} + I_{3}.\label{eq:nTorus_trace}
\end{align}
We now compute each of the above integrals separately.
\begin{itemize}
    \item {\bf Computing the integral $I_{2}$:}
    
We begin by noting that the boundary circles in the pants decomposition (see figure \ref{fig:nptTorusPants}) can be identified in the following way
\begin{align}\label{same_contours}
    \cC_{out}^{[k]}=-\cC_{in}^{[k+1]} \qquad \textrm{for} \qquad l=1,\dots,n-1.
\end{align}
With the above identification of contours, the integral $I_{2}$ in the expression \eqref{eq:nTorus_trace} simplifies as
\begin{equation}
\begin{split}
I_{2} &=\sum_{k=1}^{n}\oint_{\cC_{in}^{[k]}\cup\cC_{out}^{[k]}}\frac{\dd z}{2\pi i} \tr\left\{ \dd_{\mathcal{M}} \Psi(z) \Psi(z)^{-1}\partial_{z}\Psi(z) \Psi(z)^{-1} \right\}\\
&=\sum_{k=1}^{n}\oint_{\cC_{in}^{[k]}}\frac{\dd z}{2\pi i} \tr\left\{ \dd_{\mathcal{M}} \Psi(z) \Psi(z)^{-1}\partial_{z}\Psi(z) \Psi(z)^{-1} \right\}\\
&+\sum_{k=1}^{n}\oint_{\cC_{out}^{[k]}}\frac{\dd z}{2\pi i} \tr\left\{ \dd_{\mathcal{M}} \Psi(z) \Psi(z)^{-1}\partial_{z}\Psi(z) \Psi(z)^{-1} \right\}\\
&=\oint_{\cC_{in}^{[1]}\cup\cC_{out}^{[n]}}\frac{\dd z}{2\pi i} \tr\left\{ \dd_{\mathcal{M}} \Psi(z) \Psi(z)^{-1}\partial_{z}\Psi(z) \Psi(z)^{-1} \right\}.\label{eq:MonDer0}
\end{split}
\end{equation}
In order to express the above expression in terms of the solution to the linear problem on the torus $\Phi$, we use the identities coming from \eqref{psi_body}, \eqref{3pt_body} and fix $z_1=0$ without loss of generality:
\begin{align} \label{Phi_1n}
    \Psi(z) \vert_{\cC_{in}^{[1]}} = e^{- 2\pi i z{\bs a_{1}}} \Phi(z)\vert_{\cC_{in}^{[1]}}, && \Psi(z) \vert_{\cC_{out}^{[n]}} = e^{-2\pi i (z-\tau) {\bs a_{1}}} M_{B} \Phi(z)\vert_{\cC_{out}^{[n]}},
\end{align}
where the contours (see Figure \ref{fig:TorusParalgram_11}): 
\begin{align}
    \cC_{in}^{[1]} =\left[\frac{1-\tau}{2},-\frac{(1+\tau)}{2} \right], &&   \cC_{out}^{[n]} =\left[\frac{1+\tau}{2},-\frac{(1-\tau)}{2} \right].
\end{align}
Therefore the solution $\Phi$ restricted to the outermost circles $\cin^{[1]}$, $\cout^{[n]}$ satisfies the relation:
\begin{align}\label{Phi_outnin1}
    \Phi(z)\vert_{\cout^{[n]}} = \Phi(z+\tau)\vert_{\cin^{[1]}} \mathop{=}^{\eqref{eq:zdep_MB}} M_{B}^{-1} \Phi(z)\vert_{\cin^{[1]}} e^{2\pi i {\bs Q}}.
\end{align}
Substituting \eqref{Phi_1n} and using the identity \eqref{Phi_outnin1}, the expression \eqref{eq:MonDer0} for the integral $I_2$ simplifies as follows:
\begin{align}
I_{2}& \mathop{=}^{\eqref{Phi_1n}} \oint_{\mathcal{C}_{in}^{[1]}}\frac{\dd z}{2\pi i} \tr\left\{  \left( -2\pi i z \dd_{\mathcal{M}} {\bs a_1}  +  \dd_{\mathcal{M}}\Phi(z) \Phi(z)^{-1}  \right)   \left(\partial_{z} \Phi(z) \Phi(z)^{-1}-2\pi i\bs a_1\right)\right\} \nonumber\\
&+\oint_{\mathcal{C}_{out}^{[n]}}\frac{\dd z}{2\pi i} \tr\bigg\{\left(  -2\pi i (z-\tau) M_{B}^{-1} \dd_{\mathcal{M}} {\bs a_{1}} M_{B} + M_{B}^{-1} \dd_{\mathcal{M}} M_{B}  + \dd_{\mathcal{M}} \Phi(z) \Phi(z)^{-1} \right) \nonumber \\
& \hspace{8cm}\times \left(\partial_{z}\Phi(z) \Phi(z)^{-1}-2\pi i M_{B}^{-1} {\bs a}_1 M_{B}\right)  \bigg\} \nonumber \\
& \mathop{=}^{\eqref{Phi_outnin1}} \oint_{\mathcal{C}_{in}^{[1]}}\frac{\dd z}{2\pi i} \tr\left\{  \left( -2\pi i z \dd_{\mathcal{M}} {\bs a_1}  +  \dd_{\mathcal{M}}\Phi(z) \Phi(z)^{-1}  \right)   \left(\partial_{z} \Phi(z) \Phi(z)^{-1}-2\pi i\bs a_1\right)\right\} \nonumber\\
& - \int_{\cC_{in}^{[1]}}\frac{\dd z}{2\pi i} \tr\bigg\{\left( -2\pi iz\dd_{\mathcal{M}}\bs a_1+\dd_{\mathcal{M}}\Phi(z)\Phi(z)^{-1}+2\pi i\Phi(z)\dd_{\mathcal{M}}\bs Q\Phi(z)^{-1}  \right) \nonumber \\
& \hspace{9.5cm}\times\left(\partial_{z} \Phi(z) \Phi(z)^{-1}-2\pi i\bs a_1\right)  \bigg\} \nonumber\\
& = - \oint_{\cC_{in}^{[1]}} \dd z \tr \left\{\dd_{\mathcal{M}}\bs Q\left(\Phi(z)^{-1} \partial_{z}\Phi(z)-2\pi i\Phi(z)^{-1}\bs a_1\Phi(z)\right)\right\} 
\nonumber \\
&\mathop{=}^{\eqref{eq:nptL}}    \tr\left(\bs P\dd_{\mathcal{M}}\bs Q\right)-2\pi i\oint_{\cC_{in}^{[1]}} \dd z \tr \left\{\dd_{\mathcal{M}}\bs Q\Phi(z)^{-1}\bs a_1\Phi(z)\right\}. \label{eq:I2_step1}
\end{align}
To obtain the last line, we used the constraint \eqref{cons:CONSTRAINT}, together with the following property of $\theta_1(z)$ 
\begin{gather}
    \theta_{1}(z+1) = -\theta_{1}(z) \qquad \Rightarrow \qquad \oint_{\cC_{in}^{[1]}} \dd z \, \partial_{z} \log\theta_{1}(z)  =i\pi
\end{gather}
to simplify the contribution from the Lax matrix $L(z)$.
In summary, 
\begin{equation}\label{eq:I_2}
    I_{2} =   \tr\left(\bs P\dd_{\mathcal{M}}\bs Q\right)+2\pi i\oint_{\cC_{in}^{[1]}} \dd z \tr \left\{d_{\mathcal{M}}\bs Q\Phi(z)^{-1}\bs a_1\Phi(z)\right\}.
\end{equation}
\item {\bf Computing the integral $I_{3}$:}

Using the identification of neighbouring contours \eqref{same_contours}, the integral $I_{3}$ in \eqref{eq:nTorus_trace} simplifies as
\begin{align}
 I_{3} &= \sum_{k=1}^{n} \oint_{\cC_{in}^{[k]}\cup\cC_{out}^{[k]}}\frac{\dd z}{2\pi i} \tr\left\{ \left[\frac{\theta_1'(\bs Q-\rho)}{\theta_1(\bs Q-\rho)}-i\pi\mathbb{1}_N \right]\Psi(z)^{-1}\dd_{\mathcal{M}} \Psi(z) \right\}\nonumber \\
 & \mathop{=}^{\eqref{same_contours}}\oint_{\cC_{in}^{[1]}\cup\cC_{out}^{[n]}}\frac{\dd z}{2\pi i} \tr\left\{ \left[\frac{\theta_1'(\bs Q-\rho)}{\theta_1(\bs Q-\rho)}-i\pi\mathbb{1}_N \right]\Psi(z)^{-1}\dd_{\mathcal{M}} \Psi(z) \right\} \nonumber \\
& \mathop{=}^{\eqref{Phi_1n}, \eqref{Phi_outnin1}}-\oint_{\cC_{in}^{[1]}}\dd z\tr\left\{ \left[\frac{\theta_1'(\bs Q-\rho)}{\theta_1(\bs Q-\rho)}-i\pi\mathbb{1}_N \right]\dd_{\mathcal{M}}\bs Q \right\}=\tr\left( \frac{\theta_1'(\bs Q-\rho)}{\theta_1(\bs Q-\rho)}\dd_{\mathcal{M}}\bs Q \right) \nonumber\\
& =\dd_{\mathcal{M}}\log\left(\prod_{i=1}^N\theta_1(Q_i-\rho) \right).\label{eq:I_3}
\end{align}
The last line is obtained by remembering that ${\bs Q}\in SL(N)$ and is therefore traceless.

\item {\bf Computing the integral $I_{1}$:} 

Let us start by expressing the integral $I_{1}$ \eqref{eq:nTorus_trace} in terms of the solutions to the linear problems using \eqref{psi_body}, \eqref{3pt_body}:
\begin{align}
&I_1= -\sum_{k=1}^{n}\oint_{\cC_{in}^{[k]}\cup\cC_{out}^{[k]}}\frac{\dd z}{2\pi i} \tr\left\{ \dd_{\mathcal{M}} \Psi(z) \Psi(z)^{-1}\partial_{z}\Yt^{[k]}(z) \Yt^{[k]}(z)^{-1}\right\} \nonumber\\
& \mathop{=}^{\eqref{psi_body}, \eqref{3pt_body}} -\oint_{\mathcal{C}_{in}^{[k]}}\frac{\dd z}{2\pi i} \tr\bigg\{ \left( -2\pi i z S_{k} \dd_{\mathcal{M}} {\bs a_{k}} S_{k}^{-1} - \dd_{\mathcal{M}} S_{k} S_{k}^{-1} +\dd_{\mathcal{M}} \Phi(z) \Phi(z)^{-1} \right)\nonumber \\
& \hspace{6.0cm} \times\left(\partial_{z} \Phi_{3pt}^{[k]}(z-z_k)  \Phi_{3pt}^{[k]}(z-z_k)^{-1}-2\pi i S_{k}{\bs a}_k S_{k}^{-1}\right) \bigg\} \nonumber \\
& -\oint_{\mathcal{C}_{out}^{[k]}}\frac{\dd z}{2\pi i} \tr\bigg\{ \left( -2\pi i \left( z - \delta_{k,n} \tau \right)S_{k+1} \dd_{\mathcal{M}} {\bs a_{k+1}} S_{k+1}^{-1} - \dd_{\mathcal{M}} S_{k+1} S_{k+1}^{-1} + \dd_{\mathcal{M}} \Phi(z) \Phi(z)^{-1} \right) \nonumber \\
& \hspace{5cm} \times \left(\partial_{z} \Phi_{3pt}^{[k]}(z-z_k)  \Phi_{3pt}^{[k]}(z-z_k) ^{-1}-2\pi iS_{k+1} {\bs a}_{k+1} S_{k+1}^{-1}\right) \bigg\} \nonumber\\
& =: \sum_{k=1}^{n}\left(I_{in}^{[k]}+I_{out}^{[k]}\right)\mathop{+}^{\eqref{id:S_1n},\eqref{same_contours}}\oint_{\cC_{in}^{[1]}}\dd z\tr\left\{\left( -2\pi i z  \dd_{\mathcal{M}} {\bs a_{1}} + \dd_{\mathcal{M}} \Phi(z) \Phi(z)^{-1} \right)\bs a_1 \right\}\\
& +\oint_{\cC_{out}^{[n]}}\dd z\tr\left\{\left(- 2\pi i (z-\tau) \dd_{\mathcal{M}} {\bs a_{k}}  + \dd_{\mathcal{M}} M_B M_B^{-1} + M_B \dd_{\mathcal{M}} \Phi(z) \Phi(z)^{-1} M_B^{-1}\right)\bs a_1 \right\}\\
& =\sum_{k=1}^{n}\left(I_{in}^{[k]}+I_{out}^{[k]}\right)-2\pi i\oint_{\cC_{in}^{[1]}} \dd z \tr \left\{\dd_{\mathcal{M}}\bs Q\Phi(z)^{-1}\bs a_1\Phi(z)\right\}, \label{eq:I_1_collect}
\end{align}
where we defined
\begin{align}
I_{in}^{[k]}&:=-\oint_{\mathcal{C}_{in}^{[k]}}\frac{\dd z}{2\pi i} \tr\bigg\{ \left( -2\pi i z S_{k} \dd_{\mathcal{M}} {\bs a_{k}} S_{k}^{-1} - \dd_{\mathcal{M}} S_{k} S_{k}^{-1} + \dd_{\mathcal{M}} \Phi(z) \Phi(z)^{-1} \right)\nonumber \\
& \hspace{7cm} \times\partial_{z} \Phi_{3pt}^{[k]}(z-z_k)  \Phi_{3pt}^{[k]}(z-z_k)^{-1}\bigg\} \label{def:I_ink}
\end{align}
\begin{align}
I_{out}^{[k]}&:=-\oint_{\mathcal{C}_{out}^{[k]}}\frac{\dd z}{2\pi i} \tr\bigg\{ \left( -2\pi i (z-\delta_{k,n} \tau) S_{k+1} \dd_{\mathcal{M}} {\bs a_{k+1}} S_{k+1}^{-1} - \dd_{\mathcal{M}} S_{k+1} S_{k+1}^{-1} + \dd_{\mathcal{M}} \Phi(z) \Phi(z)^{-1} \right)\nonumber \\
&\hspace{7cm} \times \partial_{z} \Phi_{3pt}^{[k]}(z-z_k)  \Phi_{3pt}^{[k]}(z-z_k)^{-1}\bigg\}.\label{def:I_outk}
\end{align}
Note that the functions 
\begin{gather}
    \left(-2\pi i z S_{k} \dd_{\mathcal{M}} {\bs a_{k}} S_{k}^{-1} - \dd_{\mathcal{M}} S_{k} S_{k}^{-1} + \dd_{\mathcal{M}} \Phi(z) \Phi(z)^{-1} \right), \nonumber \\
    \left(- 2\pi i \left( z - \delta_{k,n} \tau \right)S_{k+1} \dd_{\mathcal{M}} {\bs a_{k+1}} S_{k+1}^{-1} - \dd_{\mathcal{M}} S_{k+1} S_{k+1}^{-1} + \dd_{\mathcal{M}} \Phi(z) \Phi(z)^{-1} \right),
\end{gather}
in the integrands above are single-valued on $\cC_{in}^{[k]} $ and $\cC_{out}^{[k]} $ respectively, but they have logarithmic branch cuts as can be seen from the local solutions \eqref{eq:LocalTorus}, \eqref{eq:Local3pt}-\eqref{eq:3ptminfty}, that make it impossible to close the integration contour.

So, we introduce the following trick: we add and subtract the integrals $\widetilde{I}_{in}^{[k]}, \widetilde{I}_{out}^{[k]}$, defined to be analogous to $I_{in}^{[k]}$, $I_{out}^{[k]}$ with the solution of the torus linear system $\Phi(z)$ in \eqref{def:I_ink}, \eqref{def:I_outk} replaced by the solution of the 3-pt linear system $\Phi_{3pt}^{[k]}$, namely
\begin{align}
\widetilde{I}_{in}^{[k]}&:=-\oint_{\cC_{in}^{[k]}}\frac{\dd z}{2\pi i}\tr\bigg\{\left(-2\pi iz S_{k}\dd_{\mathcal{M}}\bs a_{k}S_{k}^{-1}-\dd_{\mathcal{M}} S_{k}S_{k}^{-1}+\dd_{\mathcal{M}} \Phi_{3pt}^{[k]}(z-z_k)\Phi_{3pt}^{[k]}(z-z_k)^{-1}  \right)\nonumber \\
& \hspace{8.5cm}\times \partial_z\Phi_{3pt}^{[k]}(z-z_k)\Phi_{3pt}^{[k]}(z-z_k)^{-1} \bigg\},
\end{align}
\begin{align}
\widetilde{I}_{out}^{[k]}&:=-\oint_{\cC_{out}^{[k]}}\frac{\dd z}{2\pi i}\tr\bigg\{\bigg(-2\pi i\left( z - \delta_{k,n} \tau \right) S_{k+1}\dd_{\mathcal{M}}\bs a_{k+1}S_{k+1}^{-1}-\dd_{\mathcal{M}} S_{k+1}S_{k+1}^{-1} \nonumber \\
& \hspace{4.5cm}+\dd_{\mathcal{M}} \Phi_{3pt}^{[k]}(z-z_k)\Phi_{3pt}^{[k]}(z-z_k)^{-1}  \bigg) \times \partial_z\Phi_{3pt}^{[k]}(z-z_k)\Phi_{3pt}^{[k]}(z-z_k)^{-1} \bigg\}.
\end{align}
We now compute the differences
\begin{align}
I_{in}^{[k]}-\widetilde{I}_{in}^{[k]}&=-\oint_{\cC_{in}^{[k]}}\frac{\dd z}{2\pi i}\tr\bigg\{\left(\dd_{\mathcal{M}}\Phi(z)\Phi(z)^{-1}-\dd_{\mathcal{M}}\Phi_{3pt}^{[k]}(z-z_k)\Phi_{3pt}^{[k]}(z-z_k)^{-1}\right)\nonumber \\
& \hspace{6cm} \times \partial_z\Phi_{3pt}^{[k]}(z-z_k)\Phi_{3pt}^{[k]}(z-z_k)^{-1} \bigg\},
\end{align}
\begin{align}
I_{out}^{[k]}-\widetilde{I}_{out}^{[k]}&=-\oint_{\cC_{out}^{[k]}}\frac{\dd z}{2\pi i}\tr\bigg\{\left(\dd_{\mathcal{M}}\Phi(z)\Phi(z)^{-1}-\dd_{\mathcal{M}}\Phi_{3pt}^{[k]}(z-z_k)\Phi_{3pt}^{[k]}(z-z_k)^{-1}\right)\nonumber\\
& \hspace{6cm} \times \partial_z\Phi_{3pt}^{[k]}(z-z_k)\Phi_{3pt}^{[k]}(z-z_k)^{-1} \bigg\},
\end{align}
and sum the above expressions to obtain the following expression
\begin{align}
& I_{in}^{[k]}-\widetilde{I}_{in}^{[k]}+I_{out}^{[k]}-\widetilde{I}_{out}^{[k]}=-\oint_{\cC_{in}^{[k]}\cup\cC_{out}^{[k]}}\frac{\dd z}{2\pi i}\tr\bigg\{\left(\dd_{\mathcal{M}}\Phi(z)\Phi(z)^{-1}-\dd_{\mathcal{M}}\Phi_{3pt}^{[k]}(z-z_k)\Phi_{3pt}^{[k]}(z-z_k)^{-1}\right)\nonumber \\
& \hspace{10cm} \times\partial_z\Phi_{3pt}^{[k]}(z-z_k)\Phi_{3pt}^{[k]}(z-z_k)^{-1} \bigg\} \nonumber\\
& =\res_{z=z_k}\tr\left\{\left(\dd_{\mathcal{M}}\Phi(z)\Phi(z)^{-1}-\dd_{\mathcal{M}}\Phi_{3pt}^{[k]}(z)\Phi_{3pt}^{[k]}(z-z_k)^{-1}\right)\partial_z\Phi_{3pt}^{[k]}(z)\Phi_{3pt}^{[k]}(z-z_k)^{-1} \right\}. \label{I_1_step1}
\end{align}
Note the orientation of the contours of integration (see figure \ref{fig:TorusParalgram_11}) when taking the residue.
Substituting the local behaviour near $z=z_k$ of the functions $\Phi$, $\Phi_{3pt}^{[k]}$ described in \eqref{eq:LocalTorus}, \eqref{eq:Local3pt} respectively, we compute the individual terms in the residue:
\begin{align}
\dd_{\mathcal{M}}\Phi\Phi^{-1}&=\dd_{\mathcal{M}} C_{k}C_{k}^{-1}+\log(z-z_k)C_k \dd_{\mathcal{M}}\bs m_k C_k^{-1}
\nonumber \\
&+ C_k (z-z_k)^{\bs m_k}\dd_{\mathcal{M}} G_k G_k^{-1}(z-z_k)^{-\bs m_k}C_k^{-1} +\mathcal{O}((z-z_k)), \label{eq:dMPhi}\\
\dd_{\mathcal{M}}\Phi_{3pt}^{[k]} (\Phi_{3pt}^{[k]}) ^{-1}&=\dd_{\mathcal{M}} C_kC_k^{-1}+\log(z-z_k)C_k \dd_{\mathcal{M}}\bs m_k C_k^{-1}
\nonumber\\
&+C_k (z-z_k)^{\bs m_k}\dd_{\mathcal{M}} G_0^{[k]}\left(G_0^{[k]}\right)^{-1}(z-z_k)^{\bs m_k}C_k^{-1}+\mathcal{O}((z-z_k)),\label{eq:dMPhi3pt} \\
\partial_z\Phi_{3pt}^{[k]}\left(\Phi_{3pt}^{[k]} \right)^{-1}&=\frac{C_k\bs m_k C_k^{-1}}{z-z_k}+\mathcal{O}(1). \label{eq:dzPhi3pt}
\end{align}
Substituting \eqref{eq:dMPhi}-\eqref{eq:dzPhi3pt} in \eqref{I_1_step1}
\begin{equation} \label{diff_I1}
\begin{split}
& I_{in}^{[k]}-\widetilde{I}_{in}^{[k]}+I_{out}^{[k]}-\widetilde{I}_{out}^{[k]}\\
&=\res_{z=z_k}\tr\left\{\left(\dd_{\mathcal{M}}\Phi(z)\Phi(z)^{-1}-\dd_{\mathcal{M}}\Phi_{3pt}^{[k]}(z)\Phi_{3pt}^{[k]}(z)^{-1}\right)\partial_z\Phi_{3pt}^{[k]}(z)\Phi_{3pt}^{[k]}(z)^{-1} \right\} \\
&=\tr\bs m_k \dd_{\mathcal{M}} G_{k}G_{k}^{-1}-\tr\bs m_k \dd_{\mathcal{M}} G_0^{[k]}\left(G_0^{[k]}\right)^{-1}. 
\end{split}
\end{equation}
The last step to compute the integral $I_{1}$ comes from noting that the integrals $\widetilde{I}_{in,out}^{[l]}$ themselves can be evaluated explicitly using the local behavior \eqref{eq:3ptpinfty}, \eqref{eq:3ptminfty} of the 3-point solution at their 'local' $\pm i\infty$
\begin{align}\label{Itin}
\widetilde{I}_{in}^{[k]} 
&=-\oint_{\cC_{in}^{[k]}}\frac{\dd z}{2\pi i}\tr \bigg\{\left(-2\pi iz S_{k}\dd_{\mathcal{M}}\bs a_{k}S_{k}^{-1}-\dd_{\mathcal{M}} S_{k}S_{k}^{-1}+\dd_{\mathcal{M}} \Phi_{3pt}^{[k]}(z-z_k)\Phi_{3pt}^{[k]}(z-z_k)^{-1}  \right)\nonumber\\
& \hspace{9.5cm} \times\partial_z\Phi_{3pt}^{[k]}(z-z_k)\Phi_{3pt}^{[k]}(z-z_k)^{-1} \bigg\}.
\end{align}
We begin by computing the following expressions
\begin{align}
 &\dd_{\mathcal{M}} \Phi_{3pt}^{[k]}(z-z_k)\Phi_{3pt}^{[k]}(z-z_k)^{-1}\vert_{z\rightarrow- i \infty}\nonumber \\
 & \mathop{=}^{\eqref{eq:3ptpinfty}} \dd_{\mathcal{M}} S_k S_{k}^{-1} +  2\pi i (z -z_{k})S_{k}\dd_{\mathcal{M}}\bs a_{k} S_{k}^{-1} + S_{k} e^{2\pi i (z -z_{k})\bs a_{k}}\dd_{\mathcal{M}}G_{-}^{[k]} \left(G_{-}^{[k]}\right)^{-1}  e^{-2\pi i (z  -z_{k})\bs a_{k}} S_k^{-1},\label{eq:3ptM_minusiinf}
\end{align}
and similarly, the $z$-derivative term 
\begin{align}
\dd_{z} \Phi_{3pt}^{[k]}(z-z_k)\Phi_{3pt}^{[k]}(z-z_k)^{-1}\vert_{z\rightarrow -i \infty}
 &= 2\pi i S_{k} {\bs a}_{k} S_{k}^{-1}. \label{eq:3ptz_minusiinf}
 \end{align}
 Substituting \eqref{eq:3ptM_minusiinf} and \eqref{eq:3ptz_minusiinf} in the integrand of \eqref{Itin}, 
\begin{align}
&\lim_{z\rightarrow-i\infty}\tr \bigg\{\left(-2\pi iz S_{k}\dd_{\mathcal{M}}\bs a_{k}S_{k}^{-1}-\dd_{\mathcal{M}} S_{k}S_{k}^{-1}+\dd_{\mathcal{M}} \Phi_{3pt}^{[k]}(z-z_k)\Phi_{3pt}^{[k]}(z-z_k)^{-1}  \right)\nonumber\\
& \hspace{9.5cm} \times\partial_z\Phi_{3pt}^{[k]}(z-z_k)\Phi_{3pt}^{[k]}(z-z_k)^{-1} \bigg\}  \nonumber\\
&=\tr \bigg\{\left( -z_{k} 2\pi i S_{k}\dd_{\mathcal{M}}\bs a_{k} S_{k}^{-1} + S_{k} e^{2\pi i (z -z_{k})\bs a_{k}}\dd_{\mathcal{M}}G_{-}^{[k]} \left(G_{-}^{[k]}\right)^{-1}  e^{-2\pi i (z  -z_{k})\bs a_{k}} S_k^{-1} \right)\nonumber\\
& \hspace{9.5cm} \times2\pi i S_{k} {\bs a}_{k} S_{k}^{-1} \bigg\} \nonumber \\
&= (2\pi i )\tr\left(-2\pi i z_{k} {\bs a}_{k} \dd_{\mathcal{M}} {\bs a}_{k} +{\bs a}_{k} \dd_{\mathcal{M}}G_{-}^{[k]} \left(G_{-}^{[k]}\right)^{-1}\right).  \label{eq:3ptz_minusiinf1}
 \end{align}
 Substituting \eqref{eq:3ptM_minusiinf}, \eqref{eq:3ptz_minusiinf}, \eqref{eq:3ptz_minusiinf1} in \eqref{Itin} we get
 \begin{align}\label{eq:Itin1}
    \widetilde{I}_{in}^{[k]} 
&=-\oint_{\cC_{in}^{[k]}}\frac{dz}{2\pi i} (2\pi i )\tr\left(-2\pi i z_{k} {\bs a}_{k} \dd_{\mathcal{M}} {\bs a}_{k} +{\bs a}_{k} \dd_{\mathcal{M}}G_{-}^{[k]} \left(G_{-}^{[k]}\right)^{-1}\right) \nonumber \\
&=\tr\left(-2\pi i z_{k} {\bs a}_{k} \dd_{\mathcal{M}} {\bs a}_{k} +{\bs a}_{k} \dd_{\mathcal{M}}G_{-}^{[k]} \left(G_{-}^{[k]}\right)^{-1}\right).
\end{align}
The term $\widetilde{I}_{out}^{[k]}$  is computed in a similar fashion
\begin{align}\label{Itout}
\widetilde{I}_{out}^{[k]}&:=-\oint_{\cC_{out}^{[k]}}\frac{\dd z}{2\pi i}\tr\bigg\{\bigg(-2\pi i\left( z - \delta_{k,n} \tau \right) S_{k+1}\dd_{\mathcal{M}}\bs a_{k+1}S_{k+1}^{-1}-\dd_{\mathcal{M}} S_{k+1}S_{k+1}^{-1} \nonumber \\
& \hspace{4.5cm}+\dd_{\mathcal{M}} \Phi_{3pt}^{[k]}(z-z_k)\Phi_{3pt}^{[k]}(z-z_k)^{-1}  \bigg) \times \partial_z\Phi_{3pt}^{[k]}(z-z_k)\Phi_{3pt}^{[k]}(z-z_k)^{-1} \bigg\}.
\end{align}
Similar to the above computation, start by computing the individual 3-pt derivative terms
\begin{align}
    \dd_{\mathcal{M}} \Phi_{3pt}^{[k]}(z)\Phi_{3pt}^{[k]}(z)^{-1}\vert_{z\rightarrow +i \infty} 
     \mathop{=}^{\eqref{eq:3ptminfty}}& \dd_{\mathcal{M}} S_{k+1} S_{k+1}^{-1} + 2\pi i S_{k+1} \left( z -z_{k}  \right) \dd_{\mathcal{M}}\bs a_{k+1} S_{k+1}^{-1} \nonumber\\
    &+  S_{k+1}e^{2\pi i {\left( z  -z_{k}  \right)}\bs a_{k+1}} \dd_{\mathcal{M}}G_{+}^{[k]} \left( G_{+}^{[k]} \right)^{-1} e^{-2\pi i {\left( z  -z_{k}  \right)}\bs a_{k+1}} S_{k+1}^{-1}, \label{eq:3ptM_plusiinf}
\end{align}
and the $z$ derivative 
\begin{align}
    \partial_{z} \Phi_{3pt}^{[k]}(z)\Phi_{3pt}^{[k]}(z)^{-1}\vert_{z\rightarrow +i \infty} 
    &=2\pi i S_{k+1} {\bs a}_{k+1} S_{k+1}^{-1}.\label{eq:3ptz_plusiinf}
\end{align}
substituting \eqref{eq:3ptM_plusiinf}, \eqref{eq:3ptz_plusiinf} in the integrand of \eqref{Itout},
\begin{align}
&\lim_{z\rightarrow i\infty}\tr\bigg\{\bigg(-2\pi i\left( z - \delta_{k,n} \tau \right) S_{k+1}\dd_{\mathcal{M}}\bs a_{k+1}S_{k+1}^{-1}-d_{\mathcal{M}} S_{k+1}S_{k+1}^{-1} \nonumber \\
& \hspace{4.5cm}+\dd_{\mathcal{M}} \Phi_{3pt}^{[k]}(z-z_k)\Phi_{3pt}^{[k]}(z-z_k)^{-1}  \bigg) \times \partial_z\Phi_{3pt}^{[k]}(z-z_k)\Phi_{3pt}^{[k]}(z-z_k)^{-1} \bigg\}\nonumber\\
 &=\tr\bigg\{\bigg(2\pi i  \delta_{k,n} \tau  S_{k+1}\dd_{\mathcal{M}}\bs a_{k+1}S_{k+1}^{-1} - 2\pi i S_{k+1} z_{k} \dd_{\mathcal{M}}\bs a_{k+1} S_{k+1}^{-1} \nonumber \\
& +  S_{k+1}e^{2\pi i {\left( z  -z_{k}  \right)}\bs a_{k+1}} \dd_{\mathcal{M}}G_{+}^{[k]} \left( G_{+}^{[k]} \right)^{-1} e^{-2\pi i {\left( z  -z_{k}\right)}\bs a_{k+1}} S_{k+1}^{-1} \bigg) 2\pi i S_{k+1} {\bs a}_{k+1} S_{k+1}^{-1} \bigg\} \nonumber \\
    &=(2\pi i ) \tr\left( 2\pi i \tau \dd_{\mathcal{M}} {\bs a}_{1} {\bs a}_{1} - z_{k} \dd_{\mathcal{M}} {\bs a}_{k+1} {\bs a}_{k+1} + a_{k+1} \dd_{\mathcal{M}}G_{+}^{[k]} \left( G_{+}^{[k]} \right)^{-1} \right).
\end{align}
In the last line we used that ${\bs a}_{n+1} = {\bs a}_{1}$. Now the term \begin{align}
    \widetilde{I}_{out}^{[k]}&:=-\oint_{\cC_{out}^{[k]}}\frac{\dd z}{2\pi i}\tr\bigg\{\bigg(-2\pi i\left( z - \delta_{k,n} \tau \right) S_{k+1}\dd_{\mathcal{M}}\bs a_{k+1}S_{k+1}^{-1}-\dd_{\mathcal{M}} S_{k+1}S_{k+1}^{-1} \nonumber \\
& \hspace{4.5cm}+\dd_{\mathcal{M}} \Phi_{3pt}^{[k]}(z-z_k)\Phi_{3pt}^{[k]}(z-z_k)^{-1}  \bigg) \times \partial_z\Phi_{3pt}^{[k]}(z-z_k)\Phi_{3pt}^{[k]}(z-z_k)^{-1} \bigg\} \nonumber \\
&=-\oint_{\cC_{out}^{[k]}}\frac{\dd z}{2\pi i } (2\pi i ) \tr\left( 2\pi i \tau \dd_{\mathcal{M}} {\bs a}_{1} {\bs a}_{1} - z_{k} \dd_{\mathcal{M}} {\bs a}_{k+1} {\bs a}_{k+1} + a_{k+1} \dd_{\mathcal{M}}G_{+}^{[k]} \left( G_{+}^{[k]} \right)^{-1} \right) \nonumber \\
& = \tr\left( -2\pi i \tau \dd_{\mathcal{M}} {\bs a}_{1} {\bs a}_{1} + z_{k} \dd_{\mathcal{M}} {\bs a}_{k+1} {\bs a}_{k+1} - a_{k+1} \dd_{\mathcal{M}}G_{+}^{[k]} \left( G_{+}^{[k]} \right)^{-1} \right)\label{eq:Itout1}
\end{align}

Gathering the expressions \eqref{diff_I1}, \eqref{eq:Itin1}, \eqref{eq:Itout1}, the integral $I_{1}$ is
\begin{align}\label{eq:I_1}
    I_{1} &= \sum_{k=1}^{n}\bigg(\tr\bs m_k\dd_{\mathcal{M}} G_{k}G_{k}^{-1}-\tr\bs m_k\dd_{\mathcal{M}} G_0^{[k]}\left(G_0^{[k]}\right)^{-1} +\tr\bs a_k\dd_{\mathcal{M}} G_{-}^{[k]}\left(G_{-}^{[k]}\right)^{-1} \nonumber \\
    & \hspace{8.5cm}-\tr\bs a_{k+1}\dd_{\mathcal{M}} G_{+}^{[k]}\left(G_{+}^{[k]}\right)^{-1}\bigg) \nonumber \\
    & -i\pi \tr \tau \dd_{\mathcal{M}} {\bs a}_{1}^2 + \sum_{k=1}^{n} i \pi \tr z_{k} \dd_{\mathcal{M}} \left({\bs a}_{k+1}^2 - {\bs a}_{k}^2 \right) -2\pi i\oint_{\cC_{in}^{[1]}} \dd z \tr \left\{\dd_{\mathcal{M}}\bs Q\Phi(z)^{-1}\bs a_1\Phi(z)\right\}.
\end{align}
\end{itemize}
Substituting the expressions \eqref{eq:I_2}, \eqref{eq:I_3}, \eqref{eq:I_1} in \eqref{eq:nTorus_trace}, we get the expression \eqref{eq:der_Fred_mon} for the derivative of the Fredholm determinant w.r.t the monodromy data.
\end{proof}
Having computed the derivative of the Fredholm determinant, we are ready to turn to the isomonodromic tau function \eqref{eq:Thm2}.
\begin{theorem}\label{thm:TauGenFn}
The full parametric dependence of the tau function $\T_H$ is 
\begin{align}
\dd\log\T_H 
& =\omega-\omega_{3pt}, \label{eq:TauHGen}
\end{align}
where
\begin{equation}\label{eq:OmegaOmega3pt}
\omega_{3pt}:=-\sum_{k=1}^n\left(\tr\bs a_k\dd_{\mathcal{M}} G_{-}^{[k]}\left(G_{-}^{[k]}\right)^{-1} - \tr\bs m_k\dd_{\mathcal{M}} G_0^{[k]}\left(G_0^{[k]}\right)^{-1}-\tr\bs a_{k+1}\dd_{\mathcal{M}} G_{+}^{[k]}\left(G_{+}^{[k]}\right)^{-1}\right),
\end{equation}
\begin{align}\label{eq:OmegaOmega0}
	\omega=\sum_{j=1}^NP_j \dd_{\mathcal{M}} Q_j+\sum_{k=1}^n\tr\bs m_k\dd_{\mathcal{M}} G_{k}G_{k}^{-1}+\sum_{k=1}^n H_k \dd z_k+\frac{1}{2\pi i}H_\tau \dd\tau,
\end{align}
$P_{i}, Q_{i}$ are the dynamical variables in \eqref{eq:nptL}, $m_{k}$, $a_{k}$ constitute the monodromy data (see figure \ref{fig:nptTorusPants}), $\tau$ is the modular parameter, $H_{k}$, $H_{\tau}$ are the Hamiltonians \eqref{eq:IsomHams}, the matrices $G_{k}$ diagonalise the linear system on the $n$-point torus \eqref{eq:AGmGinv}, and the matrices $G_{\pm}, G_{0}$ diagonalise the 3-point linear system \eqref{eq:AGmGinv3pt}.
\end{theorem}
\begin{proof}
To compute the full parametric dependence of the isomonodromic tau function $\T_H$, we need to differentiate the prefactor in equation \eqref{eq:Thm2}:
\begin{equation}
\begin{split}
& \dd_{\mathcal{M}}\log\left(e^{-i \pi N \rho} \prod_{i=0}^N\frac{\eta(\tau)}{\theta_1\left(Q_i-\rho \right)}\prod_{k=1}^n e^{-i\pi z_k\left(\tr\bs a_{k+1}^2-\tr\bs a_{k}^2\right)}  e^{i\pi\tau\tr\left(\bs a_1^2+\frac{\mathbb{I}}{6} \right)} \right)=
\\
& =\bs-\sum_{i=1}^N\dd_{\mathcal{M}} Q_i 
\frac{\theta_1'(Q_i-\rho)}{\theta_1(Q_i-\rho)}-i\pi\sum_{k=1}^nz_k\tr(\dd_{\mathcal{M}}\bs a_{k+1}^2-\dd_{\mathcal{M}}\bs a_k^2) + i\pi \tau \dd_{\mathcal{M}}\tr {\bs a}_{1}^2.
\end{split}
\end{equation}
\eqref{eq:TauHGen} then follows from \eqref{eq:tder} and Proposition \ref{prop:derFred}.
\end{proof}
\begin{remark}
Note that $\dd\log\T_H$ is automatically a closed  1-form on the space $\mathbb{T}_{1,n}\times\mathcal{M}$, because $\T_H$ defined by equation \eqref{eq:Thm2} is a (locally) well-defined function of monodromies and times, so its partial derivatives commute.
\end{remark}
The following nontrivial statement follows from Theorem \ref{thm:TauGenFn}.
\begin{corollary}\label{thm:OmegaConst}
    The exterior derivative of the one-form $\omega$ in \eqref{eq:OmegaOmega0} is a (time-independent) two-form on $T^*{\cal M}_{1,n}$, i.e.
\begin{align}
    \dd\omega(\partial_t,\partial_{\mu_i})=\dd\omega(\partial_t,\partial_{t'})=0, && \partial_t\dd\omega\big|_{\{\mu_j\}\,\mathrm{fixed}}=0,
\end{align}
for every monodromy coordinate $\mu_i$ and time-coordinate $t$, $t'$.
\end{corollary}
\begin{proof}
Since $\omega_{3pt}$ is a time-independent one-form on $T^*\mathcal{M}_{1,n}$, $\dd\omega_{3pt}=\dd_{\cal M}\omega_{3pt}$ is a time-independent two-form on $T^*\mathcal{M}$.
On the other hand, from $\dd\log\T=0$ we have
\begin{equation}
    \dd\omega=\dd\omega_{3pt}.
\end{equation}
\end{proof}
\subsection{The closed one-form as the generating function of the Riemann-Hilbert map}\label{subsec:1pt}
The closed one-form $\dd\log\T$ has an elegant geometric interpretation as the generating function of the extended monodromy map
\begin{equation}
\mathcal{A}_{1,n} \rightarrow \mathcal{M}_{1,n}. \label{eq:mon_map_1n}
\end{equation}
Let us consider for illustration purposes the case of the one-punctured torus with the singularity at $z=0$, with the Lax pair
\begin{align}
    L_{z}^{(CM)} =\left(\begin{array}{cc}
        P(\tau) & mx(-2Q(\tau),z) \\
        mx(2Q(\tau),z) & -P(\tau)
    \end{array}\right), && L_{\tau}^{(CM)}=m\left( \begin{array}{cc}
    	0 & y(-2Q,z), \\
    	y(2Q,z) & 0
    \end{array} \right), \label{linear_systemCM}
\end{align}
where $y(\xi,z):=\partial_\xi x(\xi,z)$. The consistency condition of the above Lax matrices gives the non-autonomous elliptic Calogero-Moser equation
\begin{align}
    2\pi i\frac{\dd P}{\dd\tau}=m^2\wp'(2Q|\tau), && 2\pi i\frac{\dd Q}{\dd\tau}=P,
\end{align}
and its Hamiltonian takes the form 
\begin{gather}
    H_{CM} = P^2-m^2\wp(2Q|\tau)+4\pi im^2\partial_\tau\log\eta(\tau).
\end{gather}
The corresponding monodromy representation is
\begin{align}\label{eq:MAM0MB}
	M_A=e^{2\pi ia\sigma_3}, && M_0=C_0 e^{2\pi i m \sigma_3}C_0^{-1}, && M_B= \frac{1}{\sin 2\pi a}\left(
\begin{array}{cc}
 e^{-\frac{i \nu}{2}  }  \sin (\pi  (2 a-m)) & e^{\frac{i \nu }{2}}  \sin (\pi  m) \\
 -e^{-\frac{i \nu}{2}}  \sin (\pi  m) & e^{\frac{i \nu }{2}}  \sin (\pi  (2 a+m)) \\
\end{array}\right).
\end{align}
The spaces $\cA_{1,1}$ and $\cM_{1,1}$ are parametrerized by $m,P,Q,g,\tau$ and $m,a,\nu,c$ respectively, where $g$ is introduced below, and $c$ parametrizes the freedom of sending $C_0\mapsto C_0e^{c\sigma_3}$ in \eqref{eq:MAM0MB} without changing $M_0$. The variables $a,\nu$ are Darboux coordinates for the Goldman bracket (see Appendix \ref{sec:AppChar}), while the variables $P,Q$ are Darboux coordinates on $\mathcal{A}_{1,1}^{(0)}$. The residue of the Lax matrix
\begin{align}
	\res_{z=0}L_{CM}=-m\sigma_1=G^{-1}m\sigma_3G,\end{align}
where
\begin{align}\label{eq:GG0g}
	G:=e^{\frac{1}{2}g\sigma_3}G_0, && G_0:=\left( \begin{array}{cc}
	1 & -1 \\
	1 & 1		
	\end{array} \right).
\end{align}
The equation \eqref{eq:TauHGen} in the one-punctured case takes the simple form
\begin{align}\label{eq:Omegas1pt}
	\dd\log\T_{CM}=\omega-\omega_{3pt}, && \omega=2P\dd_{\cal M}Q+\frac{1}{2\pi i}H_\tau\dd\tau+m\dd_{\cal M}g, && \omega_{3pt}=ia\dd\nu+m \dd f,
\end{align}
where $f$ is a function only of the extended monodromy data that can be obtained by using the explicit $G$-matrices for the three-point problem from \cite{Its2016}, and whose specific form will not be needed in the following.
\begin{theorem}\label{thm:GenFnCM}
The derivative of the 1-form $\omega$ in \eqref{eq:Omegas1pt} can be written as a closed two-form on $\mathcal{A}_{1,1}$, extending the standard symplectic form\footnote{Note that \(\dd\omega\) is not quite a symplectic form. To make it symplectic one first needs to take into account that we work with extended monodromy space, so there is conjugate variable for \(m\), we call it \(c\). It can also be interpreted as initial condition for the equation \eqref{eq:gdot} and contributes as \(dm\wedge dc\). Then we may either add formal dual variable for \(\tau\), \(h\) and consider extended Hamiltonian \(\mathcal{H}=H-2\pi i h\) in order to describe non-autonomus system as autonomous one by adding equation \(\dot\tau=1\), or just restrict to extended monodromy manifold only.}
 $\dd P\wedge \dd Q$ by including variations of the Casimir $m$ and time $\tau$:
\begin{equation}
    \dd\omega=2\dd P\wedge \dd Q-\dd H_{CM}\wedge\frac{\dd\tau}{2\pi i}+\dd m\wedge \dd g.
\end{equation}
\end{theorem}
\begin{proof}
	Let us start by computing the total exterior derivative of $\omega$. Using the notation $\dot{f}:=\partial_\tau f$,
	\begin{equation}
	\begin{split}
		\dd\omega & =\dd\left(2P\dd_{\cal M}Q+H_\tau\frac{\dd\tau}{2\pi i}+m\dd_{\cal M}g \right) \\
		& =2 \dd_{\cal M}P\wedge\dd_{\cal M}Q+2\dot{P}\dd\tau\wedge\dd_{\cal M}Q+2P\dd\tau\wedge\dd_{\cal M}\dot{Q} \\
		& +\left(\frac{\partial H_{CM}}{\partial Q}\dd_{\cal M}Q+\frac{\partial H_{CM}}{\partial P}\dd_{\cal M}P+2m\wp(2Q|\tau)\dd m \right)\wedge\frac{\dd\tau}{2\pi i}+\dd m\wedge \dd_{\cal M}g+\dd\tau\wedge\dd_{\cal M}\dot{g}\\
		& =2\dd_{\cal M}P\wedge \dd_{\cal M}Q+\left(\frac{\partial H_{CM}}{\partial Q}-2\dot{P} \right)\dd_{\cal M}Q\wedge\dd\tau+\left(\frac{\partial H_{CM}}{\partial P}-4\pi i\dot{Q} \right)\dd_{\cal M}P\wedge\frac{\dd\tau}{2\pi i} \\
		& +\dd m\wedge \dd_{\cal M}g+\dd\tau\wedge\dd_{\cal M}\dot{g}-2m\wp(2Q|\tau)\dd m\wedge \frac{\dd\tau}{2\pi i} \\
		& =2\dd_{\cal M}P\wedge \dd_{\cal M}Q+\dd m\wedge \dd_{\cal M}g-4m^2\wp'(2Q|\tau)\dd_{\cal M}Q\wedge\dd\tau\\
		&+m\dd\tau\wedge\dd_{\cal M}\dot{g}-2m\wp(2Q|\tau)\dd m\wedge \frac{\dd\tau}{2\pi i},
	\end{split}
	\end{equation}
	where in the last line we used the Hamiltonian equations $2\pi i\dot{Q}=P$, $2\pi i\dot{P}=m^2\wp'(2Q)$. To go further, we need to compute $\dot{g}$. Consider the local behaviour at zero of the Lax equation $2\pi i\partial_\tau\Phi=\Phi L_{\tau,CM}$. The LHS and RHS behave respectively as
	\begin{align}
		2\pi i\partial_\tau\Phi\sim Cz^{m\sigma_3}\partial_\tau G, && \Phi L_\tau\sim C z^{m\sigma_3}G\lim_{z\rightarrow 0}L_{\tau,CM}(z,\tau),
	\end{align}
	implying that
	\begin{equation}
		2\pi i G^{-1}\partial_\tau G=i\pi G_0^{-1}\partial_\tau g\sigma_3G_0=\lim_{z\rightarrow0}L_{\tau,CM}(z,\tau)=-m\wp(2Q|\tau)\sigma_1,
	\end{equation}
	i.e.
	\begin{equation}\label{eq:gdot}
		2\pi i\dot{g}=-2m\wp(2Q|\tau).
	\end{equation}
	Plugging this in the expression for $\dd\omega$, we find
	\begin{equation}
		\begin{split}
			\dd\omega & =2\dd_{\cal M}P\wedge \dd_{\cal M}Q+\dd m\wedge \dd_{\cal M}g-4m^2\wp'(2Q|\tau)\dd_{\cal M}Q\wedge\dd\tau\\
		&-m\frac{\dd\tau}{2\pi i}\wedge\dd_{\cal M}\left(2m\wp(2Q|\tau) \right)-2m\wp(2Q|\tau)\dd m\wedge \frac{\dd\tau}{2\pi i} \\
		& =2\dd_{\cal M}P\wedge\dd_{\cal M}Q+\dd m\wedge\dd_{\cal M}g.
		\end{split}
	\end{equation}
	This makes it clear that $\dd\omega$ has no $\dd\tau$-component, as implied by Corollary \ref{thm:OmegaConst}.\footnote{One can also explicitly check its time-independence:
\begin{equation}
\begin{split}
	2\pi i\partial_\tau\dd\omega & =2\dd_{\cal M}\dot{P}\wedge \dd_{\cal M}Q+2\dd_{\cal M}P\wedge \dd_{\cal M}\dot{Q}+\dd m\wedge \dd_{\cal M}\dot{g} \\
	& = 2\dd_{\cal M}\left(m^2\wp'(2Q|\tau) \right)\wedge\dd_{\cal M}Q+2\dd_{\cal M}P\wedge\dd_{\cal M}P-\dd m\wedge\dd_{\cal M}\left(2m\wp(2Q|\tau) \right)=0.
\end{split}
\end{equation}}
We now rewrite $\dd\omega$ as a 2-form on on $\mathcal{A}_{1,1}$:
	\begin{equation}
	\begin{split}
	    2\dd_{\cal M}P\wedge \dd_{\cal M}Q & =2\dd P\wedge\dd Q-2\dot{P}\dd\tau\wedge\dd_{\cal M}Q+2\dot{Q}\dd\tau\wedge\dd_{\cal M}P \\
	    & =2\dd P\wedge \dd Q+\frac{\dd\tau}{2\pi i}\wedge\left(-2m^2\wp
	    '(2Q|\tau)\dd_{\cal M}Q+2P\dd_{\cal M}P \right)\\
	    & =2\dd P\wedge\dd Q-\dd H\wedge \frac{\dd\tau}{2\pi i}-\dd_t g\wedge \dd m,
	\end{split}
	\end{equation}
	leading to
	\begin{equation}
	\begin{split}
	    \dd\omega &=2\dd P\wedge\dd Q-\dd H\wedge \frac{\dd\tau}{2\pi i}-\dd_t g\wedge \dd m+\dd m\wedge\dd_{\cal M}g\\
	    & =2\dd P\wedge\dd Q+\dd m\wedge\dd g-\dd H\wedge \frac{\dd\tau}{2\pi i}.
	\end{split}
	\end{equation}
\end{proof}

\begin{corollary}
    The tau function $\T_{CM}$ is the generating function for the extended monodromy map \eqref{eq:mon_map_1n} on the one-punctured torus, i.e. it is the difference of symplectic potentials on $\mathcal{A}_{1,1}$ and $\mathcal{M}_{1,1}$ respectively.
\end{corollary}
\begin{proof}
This follows from the equation $\dd\log\T=\omega-\omega_{3pt}$ together with the following two facts:
\begin{enumerate}
    \item Theorem \ref{eq:Thm2}, stating that $\omega$ is a symplectic potential on $\mathcal{A}_{1,1}$.
    \item $\omega_{3pt}$ is a symplectic potential for the extended Goldman's symplectic form \cite{Alekseev1995}: its exterior derivative is a closed 2-form on the character variety such that its restriction on the symplectic leaves yields the Goldman symplectic form:
\begin{equation}
	\dd\omega_{3pt}\big|_{\dd m=0}=i\dd a\wedge \dd\nu \mathop{=}^{\eqref{eq:OmegaGApp}} \frac{i}{2\pi } \Omega_{G}. \label{sym:goldman}
\end{equation}
\end{enumerate}
\end{proof}

This statement is readily generalized to the $SL(N)$ case with arbitrary number of punctures by using the explicit one-forms and \eqref{eq:OmegaOmega0} and \eqref{eq:OmegaOmega3pt}. Among its consequences is the extension of the point of view of \cite{Bertola2019}, identifying the isomonodromic tau function with the generating function of the extended monodromy map, to the case of $SL(N)$ Fuchsian systems on elliptic curves. Indeed, the existing results for Fuchsian systems on the Riemann sphere can be derived directly from the Fredholm determinant representation of the tau function, in much the same way as we did in the genus one case, see Appendix \ref{sec:nSphere}.

\section{Concluding remarks}
The genus zero analogue of our one-form $\dd\log\T_H$ computed in Theorem \ref{thm:TauGenFn} was used in \cite{Its2016} to compute the connection constant for the Painlev\'e VI and II equation, which is the proportionality constant between tau functions in different asymptotic regimes. In the case of Painlev\'e VI, it corresponds to the transformation $t\mapsto(1-t)$, which belongs to the modular group of the four-punctured sphere. In the case of the torus, there is a new type of connection constant corresponding to an S-duality transformation $\tau\mapsto -1/\tau\in SL(2,\mathbb{Z})$, the modular group of the torus.  The problem of computing such constants has been traditionally an outstanding problem in the theory of Painlev\'e equations, and we defer the explicit computation of the modular connection constant to our upcoming paper. Such connection constants can be viewed as originating from a change in monodromy coordinates induced by a change in the pants decomposition (for example, viewing the torus as a pair of pants glued along the A- or B-cycle, in the case of $\tau\mapsto-1/\tau$).

 The higher genus generalization of the Lax matrix \eqref{eq:nptL} is a twisted meromorphic differential on a punctured Riemann surface of genus $g$  \cite{Levin1999,Levin2013,DelMonte2020thesis} , with periodicity properties along the A- and B-cycles given by twist matrices $T_{A_j},\,T_{B_j}\in SL(N,\mathbb{C})$:
\begin{align}
	L_z(\gamma_{A_j}z)=T_{A_j}^{-1}L_z(z)T_{A_j}, && L_z(\gamma_{B_j}z)=T_{B_j}^{-1}L_z(z)T_{B_j}, && \prod_{j=1}^{g}T_{A_j}T_{B_j}T_{A_j}^{-1}T_{B_j}^{-1}=\mathbb{1}_N.
\end{align}
As usual, $L_z(z)$ will have $n$ simple poles, with residues \eqref{eq:AGmGinv}. The twist matrices encode the canonical variables of the isomonodromic system, and different types of twists will correspond to different flat bundles on $C_{g,n}$. 

The structure \eqref{eq:TauHGen} of the closed 1-form $\omega-\omega_0$ has a rather straightforward generalization, leading to natural expectations on what should happen in the case of a punctured Riemann surface of genus $g\ge 2$. The expectation is that a tau function in that case should satisfy
\begin{equation}\label{eq:HigherGenusOneForm}
	\dd\log\T_H^{(g,n)}=\sum_{a=1}^{(N^2-1)(g-1)}P_a\dd_{\cal M}Q_a+\sum_{j=1}^{3g-3+n}H_j\dd t_j+\sum_{k=1}^n\tr\bs m_kd_{\mathcal{M}} G_{k}G_{k}^{-1}-\omega_{3pt}.
\end{equation}
Here we denoted by $P_a,Q_a$ an appropriate set of Darboux coordinates encoded in the twists, and the $t_j$'s are local coordinates on the Teichm\"uller space $\mathbb{T}_{g,n}$. $\omega_0$ is as before a time-independent one-form depending on the pants decomposition.

While the extension \eqref{eq:HigherGenusOneForm} seems very natural, its derivation would require first the generalization to higher genus Riemann surfaces of the Fredholm determinant construction of \cite{DelMonte2020}, which we leave to future work.

\begin{appendix}	
\section{Character variety for the one-punctured torus}\label{sec:AppChar}
The character variety of the $SL(2)$ one-punctured torus is the cubic surface
\begin{equation}
\mathcal{M}_{1,1}^{(0)}:=\left\{M_A,M_B\in SL(2):\, \tr\left(M_A^{-1}M_B^{-1}M_A M_B \right)=M_0 \right\}/\sim.
\end{equation}
The trace coordinates
\begin{align}
p_A:=\tr M_A, && p_B:=\tr M_B, && p_{AB}:=\tr M_AM_B, && p_0:=\tr M_0
\end{align}
satisfy the relation
\begin{equation}
\begin{split}
p_0=\tr (M_A^{-1}M_B^{-1}M_AM_B)&=\tr (M_AM_BM_A^{-1})\tr (M_B)-\tr (M_AM_BM_A^{-1}M_B) \\
&=(\tr M_B)^2-\tr(M_AM_B)\tr(M_AM_B^{-1})+\tr(M_A^2) \\
&=(\tr M_B)^2-\tr(M_AM_B)\left[\tr (M_A)\tr (M_B)-\tr(M_AM_B) \right]+\tr(M_A)^2-2 \\
&=p_A^2+p_B^2+p_{AB}^2-p_Ap_Bp_{AB}-2,
\end{split}\label{cubic:1pt}
\end{equation}
so that the space of monodromy data is described by the cubic surface
\begin{equation}
\mathcal{W}_{1,1}=p_A^2+p_B^2+p_{AB}^2-p_Ap_Bp_{AB}-p_0-2=0
\end{equation}
known as Fricke cubic \cite{Fricke1897}. We used the following identities to obtain \eqref{cubic:1pt}:
\begin{align}\label{eq:SL2id}
\tr(xy)+\tr(xy^{-1})=\tr x\tr y, && \tr(x^2)=(\tr x)^2-2, &&\tr(x^{-1})=\tr(x).
\end{align}
\subsection{Darboux Coordinates}
The functions $p_A,p_B,p_{AB}$ satisfy the Goldman bracket:
\begin{equation}\label{eq:GoldmanPoisson}
\{p_A,p_B\}=p_{AB}-\frac{1}{2}p_Ap_B.
\end{equation}
We now introduce Darboux coordinates for the space $\mathcal{M}^{(0)}_{1,1}$. One possible choice is the one presented in \cite{Nekrasov2011}, that involves hyperbolic functions and square-roots. A more convenient choice, involving only trigonometric functions, is explicit in the parametrization \eqref{eq:MAM0MB}. In terms of $a,\nu,m$, we have
\begin{gather}
p_A=2\cos2\pi a, \nonumber \\
p_B=\frac{\sin(\pi(2a-m))}{\sin2\pi a}e^{-i\nu/2}+\frac{\sin(\pi(2a+m))}{\sin2\pi a}e^{i\nu/2}, \nonumber, \\
p_{AB}=\frac{\sin(\pi(2a-m))}{\sin2\pi a}e^{i(2\pi a-\nu/2)}+\frac{\sin(\pi(2a+m))}{\sin2\pi a}e^{-i(2\pi a-\nu/2)}, \nonumber\\
p_0=2\cos2\pi m.\label{eq:pcoords}
\end{gather}
From the relations \eqref{eq:pcoords}, we see that $a,\nu$ are Darboux coordinates for the Goldman bracket \eqref{eq:GoldmanPoisson}, and the 
Goldman's symplectic form is
\begin{equation}\label{eq:OmegaGApp}
\Omega_G=2\pi \dd a\wedge \dd\nu.
\end{equation}
\section{Tau function for the $n$-punctured sphere}\label{sec:nSphere}
The construction of Section \ref{sec:MainThms} applies (in a much simpler context) also to the case of an isomonodromic problem on the Riemann sphere with $n$ Fuchsian singularities. The tau function in this case was written as a Fredholm determinant in \cite{Gavrylenko2016b,Cafasso2017}:
\begin{equation}
	\T_H^{(g=0)}= \det_{\cH_{+}}\left[\cP_{\Sigma,+}^{-1} \cP_{\oplus,+} \right] \prod_{k=1}^{n-2}z_k^{\frac{1}{2}\tr\left(\bs a_k^2-\bs a_{k-1}^2-\bs m_k^2 \right)}, \label{eq:GL}
\end{equation}
where $\mathcal{P}_{\Sigma,+}$ and $\mathcal{P}_{\oplus,+}$ are here defined from the $SL(N)$ Fuchsian linear system on a sphere and the corresponding pants decomposition as in Figure \ref{fig:nptSpherePants} (see \cite{Gavrylenko2016b} for details).
\begin{figure}[h]
\begin{center}
\begin{tikzpicture}[scale=1]

\draw[thick,decoration={markings, mark=at position 0.75 with {\arrow{>}}}, postaction={decorate}] (-5,1.5) circle[x radius=0.5, y radius=0.2];
\draw(-6,0.5) to[out=0,in=270] (-5.5,1.5);
\draw(-4,0.5) to[out=180,in=270] (-4.5,1.5);
\draw(-6,-0.5) to[out=0,in=180] (-4,-0.5);
\draw[thick, red,decoration={markings, mark=at position 0.6 with {\arrow{>}}}, postaction={decorate}]  (-4,0.5) to[out=10,in=-10] (-4,-0.5);
\draw[thick,dashed,red] (-4,0.5) to[out=220,in=-220] (-4,-0.5);
\draw[thick,decoration={markings, mark=at position 0 with {\arrow{>}}}, postaction={decorate}] (-6,0) circle[y radius=0.5, x radius=0.2];

\draw[thick,blue,decoration={markings, mark=at position 0 with {\arrow{<}}}, postaction={decorate}] (-3.3,0) circle[y radius=0.5, x radius =0.2];
\draw (-3.3,-0.5) to (-2.7,-0.5);
\draw (-3.3,0.5) to (-2.7,0.5);
\draw[thick, blue,decoration={markings, mark=at position 0.6 with {\arrow{<}}}, postaction={decorate}]  (-2.7,0.5) to[out=10,in=-10] (-2.7,-0.5);
\draw[dashed, blue] (-2.7,0.5) to[out=220,in=-220] (-2.7,-0.5);

\draw[thick,decoration={markings, mark=at position 0.75 with {\arrow{>}}}, postaction={decorate}] (-1,1.5) circle[x radius=0.5, y radius=0.2];
\draw(-2,0.5) to[out=0,in=270] (-1.5,1.5);
\draw(0,0.5) to[out=180,in=270] (-0.5,1.5);
\draw(-2,-0.5) to[out=0,in=180] (0,-0.5);
\draw[thick,red,decoration={markings, mark=at position 0 with {\arrow{>}}}, postaction={decorate}] (-2,0) circle[y radius=0.5, x radius=0.2];
\draw[thick, red,decoration={markings, mark=at position 0.6 with {\arrow{>}}}, postaction={decorate}]  (0,0.5) to[out=10,in=-10] (0,-0.5);
\draw[thick,dashed, red] (0,0.5) to[out=220,in=-220] (0,-0.5);

\draw[thick,blue,decoration={markings, mark=at position 0 with {\arrow{<}}}, postaction={decorate}] (0.7,0) circle[y radius=0.5, x radius =0.2];
\draw (0.7,-0.5) to (1.3,-0.5);
\draw (0.7,0.5) to (1.3,0.5);
\draw[thick, blue,decoration={markings, mark=at position 0.6 with {\arrow{<}}}, postaction={decorate}]  (1.3,0.5) to[out=10,in=-10] (1.3,-0.5);
\draw[dashed, blue] (1.3,0.5) to[out=220,in=-220] (1.3,-0.5);

\node at (2.3,0) {{\Large\dots}};

\draw[thick,blue,decoration={markings, mark=at position 0 with {\arrow{<}}}, postaction={decorate}] (3,0) circle[y radius=0.5, x radius =0.2];
\draw (3,-0.5) to (3.6,-0.5);
\draw (3,0.5) to (3.6,0.5);
\draw[thick, blue,decoration={markings, mark=at position 0.6 with {\arrow{<}}}, postaction={decorate}]  (3.6,0.5) to[out=10,in=-10] (3.6,-0.5);
\draw[dashed, blue] (3.6,0.5) to[out=220,in=-220] (3.6,-0.5);

\draw[thick,decoration={markings, mark=at position 0.75 with {\arrow{>}}}, postaction={decorate}] (5.3,1.5) circle[x radius=0.5, y radius=0.2];
\draw(4.3,0.5) to[out=0,in=270] (4.8,1.5);
\draw(6.3,0.5) to[out=180,in=270] (5.8,1.5);
\draw(4.3,-0.5) to[out=0,in=180] (6.3,-0.5);
\draw[thick,red,decoration={markings, mark=at position 0 with {\arrow{>}}}, postaction={decorate}] (4.3,0) circle[y radius=0.5, x radius=0.2];
\draw[thick,decoration={markings, mark=at position 0.6 with {\arrow{>}}}, postaction={decorate}]  (6.3,0.5) to[out=10,in=-10] (6.3,-0.5);
\draw[thick,dashed] (6.3,0.5) to[out=220,in=-220] (6.3,-0.5);


\node at (-5,0.2) {{\Large$\mathscr{T}^{[1]}$}};
\node at (-1,0.2) {{\Large$\mathscr{T}^{[2]}$}};
\node at (5.3,0.2) {{\Large$\mathscr{T}^{[n-2]}$}};

\end{tikzpicture}
\end{center}
\caption{Pants decomposition for the $n$-punctured sphere}\label{fig:nptSpherePants}
\end{figure}
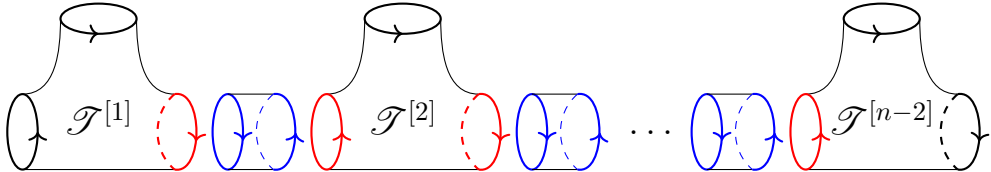

The computation of the monodromy derivative of \eqref{eq:GL} is a simplified version of the proof of Theorem \ref{thm:TauGenFn}: the result is
\begin{align}
\dd\log\T_H^{(g=0)} &=\sum_{k=1}^n\tr\bs m_k\dd_{\mathcal{M}} G_{k}G_{k}^{-1}+\sum_{k=1}^n H_k dz_k-\sum_{k=1}^n\Theta_{k}^{(3pt)}\\
& = \sum_{k=1}^n\tr\bs m_k\dd G_{k}G_{k}^{-1}-\sum_{k=1}^n H_k dz_k-\sum_{k=1}^n\Theta_{k}^{(3pt)},\label{eq:TauHGeng0}
\end{align}
where $H_k$ are the Schlesinger Hamiltonians, and
\begin{equation}
\Theta_k^{3pt}:=-\left(\tr\bs a_kd_{\mathcal{M}} G_{-}^{[k]}\left(G_{-}^{[k]}\right)^{-1} - \tr\bs m_kd_{\mathcal{M}} G_0^{[k]}\left(G_0^{[k]}\right)^{-1}-\tr\bs a_{k+1}d_{\mathcal{M}} G_{+}^{[k]}\left(G_{+}^{[k]}\right)^{-1}\right)
\end{equation}
are the one-forms coming from the pants decomposition, like in the case of the torus. In particular, this shows that the Fredholm determinant tau function of \cite{Gavrylenko2016b,Cafasso2017} is indeed the fully normalized tau function of \cite{Its2016}, and thus coincides with the generating function of the monodromy map for the $n$-punctured sphere with Fuchsian singularities, as in \cite{Bertola2019}.
\end{appendix}

\bibliographystyle{alph}
\bibliography{Biblio}

\end{document}